\documentclass[twoside,11pt,titlepage]{amsart}
\usepackage{amsmath} 
\usepackage{multirow}
\makeatletter
\newcommand{\addresseshere}{%
  \enddoc@text\let\enddoc@text\relax
}
\makeatother

\usepackage{amsthm} 
\usepackage{amssymb}	
\usepackage{graphicx} 
\usepackage{multicol} 
\usepackage[marginratio=1:1,height=584pt,width=360pt,tmargin=117pt]{geometry}
\usepackage{hyperref}
\usepackage{pst-node}
\usepackage{nccmath}

\usepackage[mathscr]{euscript}
\usepackage[numbers]{natbib}
\usepackage[toc,page]{appendix}
\usepackage{dsfont}
\author{Michael Freedman, Modjtaba Shokrian Zini, Zhenghan Wang}
\title{Quantum Computing with octonions}
\makeatletter
\newcommand\Author{Michael Freedman, Modjtaba Shokrian-Zini, Zhenghan Wang}
\let\Title\@title
\def\ps@mystyle{%
      \let\@oddfoot\@empty\let\@evenfoot\@empty
      \def\@evenhead{\makebox[0pt][l]{\thepage}\hfill\Author\hfill}%
      \def\@oddhead{\hfill\Title\hfill\makebox[0pt][l]{\thepage}}%
      \let\@mkboth\markboth}
\makeatother
\makeatletter 
\g@addto@macro{\endabstract}{\@setabstract}
\newcommand{\authorfootnotes}{\renewcommand\thefootnote{\@fnsymbol\c@footnote}}%
\makeatother
\makeatletter
\renewcommand{\maketitle} 
{ \begingroup \vskip 10pt \begin{center} \large {\bf \@title}
	\vskip 10pt \large \@author \hskip 20pt \@date \end{center}
  \vskip 10pt \endgroup \setcounter{footnote}{0} }
\makeatother 

\newcommand{\ket}[1]{\left| #1 \right>} 
\newcommand{\bra}[1]{\left< #1 \right|} 
\let\baraccent=\= 
\renewcommand{\=}[1]{\stackrel{#1}{=}} 

\newtheorem{thm}{Theorem}[section]

\newtheorem{lem}[thm]{Lemma}

\newtheorem{cor}[thm]{Corollary}

\newtheorem{ques}[thm]{Question}
\theoremstyle{definition}
\newtheorem{dfn}[thm]{Definition}
\theoremstyle{remark}
\newtheorem{rmk}[thm]{Remark}

\newcommand{\mbbC}{\mathbb{C}}

\usepackage{url,times,wasysym,geometry,indentfirst,rotating,tikz}
\title{\LARGE{\bf
\textsc{Quantum Computing with Octonions
}}}

\begin{document}
\begin{center}
  \LARGE 
   \maketitle \par \bigskip

  \normalsize
  \authorfootnotes
  Michael Freedman \footnote{michaelf@microsoft.com}\textsuperscript{,\hyperref[1]{1}},
  Modjtaba Shokrian-Zini \footnote{shokrian@math.ucsb.edu}\textsuperscript{,\hyperref[3]{3}}, Zhenghan Wang \footnote{zhenghwa@microsoft.com;
  zhenghwa@math.ucsb.edu}\textsuperscript{,\hyperref[2]{2}} \par 
  \bigskip
\end{center}
\begin{abstract}
There are two schools of ``measurement-only quantum computation''. The first (\cite{raussendorf2001one}) using prepared entanglement (cluster states) and the second (\cite{bonderson2008measurement}) using collections of anyons, which according to how they were produced, also have an entanglement pattern. We abstract the common principle behind both approaches and find the notion of a graph or even continuous family of \textbf{equiangular} projections. This notion is the leading character in the paper.  The largest continuous family, in a sense made precise in Corollary \ref{R^2n_classification_equiangular}, is associated with the \textit{octonions} and this example leads to a universal computational scheme.  Adiabatic quantum computation also fits into this rubric as a limiting case: nearby projections are nearly equiangular, so as a gapped ground state space is slowly varied the corrections to unitarity are small.
\end{abstract}
\tableofcontents
\section{Introduction}
Five hundred years ago it was a matter of scientific debate whether a boat could sail into the wind. It was a question with applications to commerce and warfare and relevant to the famous encounter in 1588 between the English fleet and the Spanish Armada. Even without the Bernoulli effect, sailing into the wind can be explained as the composition of two projections  rotating (and shrinking) a force vector. Treating the sail $S$ as a plane the wind vector $v$ is projected to $S^\perp(v)$, the keel $K$, another plane, further projects the force to $K$ where it may be written as $K(S^\perp(v))$. Mathematically this vector can easily be at an oblique angle to $v$. Ships can sail into the wind, Q.E.D. 

Five hundred years on, we are trying to build quantum computers and it is again relevant what transformations can be wrought by compositions of projections, which quantum mechanically represent the consequence of making a measurement. Now the situation is more subtle because in quantum mechanics projection is to an eigenspace of an observable, and which eigenspace is probabilistic. The observation which begins this investigation is that quantum computing to remain unitary must  not leak quantum information into the environment (we later relax this condition to consider some minimal leakage). This imposes a stringent geometric condition on which projection may follow another. The condition is \textit{equiangularity} which we define below. In higher dimensions two $k$-planes in $n$-space intersect with an ordered list of a family of $k$ dihedral angles (real and complex cases are similar). If a vector $\ket{\psi}$ in the first $k$-plane is probabilistically projected (in accordance to the rules of quantum mechanics) to the second $k$-plane or its perpendicular space, all these dihedral angles must be equal to avoid learning some statistical information about $\ket{\psi}$. To see the key point consider the non-equiangular case. In that case, if a state $\ket{\psi}$ lies in a subspace $Q$ and we observe that its projection has fallen into $P$ rather than $P^\perp$, we will rightly suspect that $\ket{\psi}$ made one of the smaller possible angles with $P$. Our a posteriori distribution will be updated from the prior. This is \textit{leakage}. The no-leakage constraint will be formulated and explained in detail.

The paper is structured as follows. In section 2, we examine the equivalence of the information theoretic condition no-leakage and the linear algebraic property of equiangularity. We then restrict ourselves to the case we call strong equiangular pairs, where $P,Q$ are strong equiangular if equiangularity also holds for their complement $P^\perp,Q^\perp$. In section 4, we partially characterize the continuous families of such pairs (see Figure \ref{strongly_equiangular_plot} and \ref{larger_plot}). The characterization is followed by the explicit construction of strong equiangular families, one of which (related to octonions), allows us to build any local unitary gate efficiently. This allows us to build a (universal) computational model in which strong equiangularity is manifest from the outset. 

In section 6, we introduce the abstract framework unifying the core aspects of famous measurement-only models. We observe that equiangularity is a core aspect of all these models. Hence, strong equiangularity should guide any implementation of a measurement based model. It needs to be noted that details of such models (like the universality of, say, measurement only topological model) require their own special tricks, as some have assumptions on which measurements are permitted. But the abstract model in section 6 based on equiangular projections underlies all earlier implementations. Section 7 involves some connections to other topics and ways to obtain new classes of equiangular projections. 

Finally, in section 8, we discuss the origins of this work. Topological protection has become a major theme, but we are interested in other forms of protection of operations and speculate that small molecules may provide a form of \textit{chemical} protection insofar as symmetries provide the rigid structure of representation spaces. In such a paradigm molecular binding acts as measurement. As mentioned, we try to identify in this paper the abstract framework and the components that are needed (equiangular projections being the main one) as a guide to building any measurement-based model in the future.

\section{No-leakage = Equiangularity}\label{No-leakage = Equiangular_section}

We start with a more algebraic point of view on equiangularity, followed by the geometric definition, and then show the equivalence of the two formulations.

In the sequence of projections carried out in a ``forced measurement protocol'', as in \cite{bonderson2008measurement}, each consecutive pair must satisfy a certain property ensuring the ability to retry a prior measurement that did not give the desired outcome. This means no information should leak to the environment, or equivalently, one should not be able to infer anything about the quantum state after a projective measurement of the form $\{P,I-P\}$ other than whether the state is now in the subspace $P$ or $I-P$.

More precisely, no leakage of information is equivalent to reversibility of the operation. In quantum mechanics, a unitary map ensures reversibility. As one performs a measurement which causes the state to move from $Q$ to $P$ and back to $Q$, no-leakage requires that the state be changed by a map which is proportional to a unitary, with positive scale in $[0,1]$ (the scale is there as the projections inevitably decrease the norm). Let us call such pairs of projections the \textit{no-leakage} pairs. In fact for the shortest of such loops $PQP$, the corresponding unitary is the identity map.

No-leakage condition implies that sequences of projections $QP_kP_{k-1}\ldots P_1Q$ of consecutive no-leakage pairs give a unitary transformation up to some scale of the states inside $Q$. The abundance of the resulting unitaries will provide enough gates to perform universal quantum computation.

In this paper, projections are always hermitian (or symmetric if on $\mathbb{R}^n$), and will be referred to by the same notation as the corresponding subspace.  Hilbert spaces in this paper are always finitely dimensional.
\begin{dfn}
Two subspaces $P,Q$ of a Hilbert space $H$ are equiangular if the minimum of $\arccos|(r_P,r_Q)|$, where $(\cdot,\cdot)$ is the inner product, over all unit vectors $r_P$ in $P$, is independent of $r_Q$, and similarly if $r_P$ is fixed instead.
\end{dfn}
How to compute the minimum? 
\begin{lem}
\begin{align}
\min_{r_P} \arccos|(r_P,r_Q)| = \arccos \frac{|(Pr_Q,r_Q)|}{||Pr_Q||}=\arccos(||Pr_Q||)
\end{align}
\end{lem}
\begin{proof}
Write $r_Q=Pr_Q+(I-P)r_Q$, then $(r_P,r_Q)=(r_P,Pr_Q)$. The minimum of the $\arccos$ above corresponds to the maximum of the absolute value of the inner product and clearly $|(r_P,Pr_Q)| \le ||Pr_Q||$.
\end{proof}
The angle above $\theta_{P,Q}$ will be called the \textit{dihedral angle} between $P,Q$. For equiangular projections, we have the following:
\begin{thm}\label{equiangularity_PQP}
Equiangularity of $P,Q$ implies $PQP=\alpha^2 P$ and $QPQ=\alpha^2 Q$, where $0<\alpha=\cos(\theta_{P,Q})<1$, meaning they act as a scalar on the image of each other.
\end{thm}
\begin{proof}
Equiangular means the norms $||Qr_P||,||Pr_Q||$ are constants. Assume $Q$ has rank $d$. By a unitary transformation, which preserves inner product, we diagonalize $Q$. $Q,P$ become: 
\begin{align}\label{P&Q_matrix_form}
Q=
\begin{pmatrix}
I_{d \times d} & 0 \\
0 & 0
\end{pmatrix},P= 
\begin{pmatrix}
U_{d \times d} & B \\
B^\dagger & D
\end{pmatrix}
\end{align}
Then
\begin{align}\label{r_Q_to_v}
\frac{|(Pr_Q,r_Q)|}{||Pr_Q||} = \frac{|(PQv,Qv)|}{||PQv||.||Qv||}
\end{align}
for any $v$ in the Hilbert space, as $\frac{Qv}{||Qv||} \in Q$ and equal to some $r_Q$. The above can be rewritten as
\begin{align}\label{v_to_v_Q}
\frac{|(PQv,Qv)|}{||PQv||.||Qv||}=\frac{|(Qv,QPQv)|}{||PQv||.||Qv||}=
\frac{|(v_Q,Uv_Q)|}{(||Uv_Q||^2+||B^\dagger v_Q||^2)^{\frac{1}{2}}||v_Q||},
\end{align}
where $v_Q=Qv$. Notice $P$ is a projection, and $P^2=P,P^\dagger=P$ imply:
\begin{align}\label{I-P=P^2}
U=U^\dagger, \ D=D^\dagger, \\\label{2-P=P^2}
BB^\dagger+UU^\dagger=U, B^\dagger B+DD^\dagger=D.
\end{align}
The second line implies $||B^\dagger v_Q||^2=(B^\dagger v_Q,B^\dagger v_Q)=(v_Q,BB^\dagger v_Q)=(v_Q,(U-UU^\dagger)v_Q)$. Hence the denominator becomes
\begin{align}\label{U_condition}
\frac{|(v_Q,Uv_Q)|}{(||Uv_Q||^2+||B^\dagger v_Q||^2)^{\frac{1}{2}}||v_Q||}=\frac{|(v_Q,Uv_Q)|}{|(v_Q,Uv_Q)|^\frac{1}{2}.||v_Q||}=\frac{|(v_Q,Uv_Q)|^\frac{1}{2}}{||v_Q||}
\end{align}
The above has to be some constant $\alpha=\cos(\theta_{P,Q})$. Since $U$ is a hermitian matrix, we can consider the unit eigenvectors of $U$ called $v_i$ with real eigenvalues $\lambda_i$ for $1 \le i \le d$. We notice that by (\ref{I-P=P^2},\ref{2-P=P^2}), $U-U^2$ is a positive matrix, meaning that $\lambda_i-\lambda_i^2\ge 0 \implies \lambda_i \ge 0$. In turn, by (\ref{U_condition}), choosing $v_Q=v_i$ gives us $\lambda_i=\alpha^2$, which means all $\lambda_i$ are equal and $U = \alpha^2 I_{d\times d}$ for some $\alpha>0$. Hence $QPQ=\alpha^2 Q$. Similarly, as the equiangularity condition is symmetric for $P,Q$ we get $PQP=\beta^2P$ for some $\beta>0$. Using (\ref{P&Q_matrix_form}) to calculate explicitly $PQP$, one obtains $\alpha=\beta$; in particular, $\theta_{P,Q}=\theta_{Q,P}$ as expected.
\end{proof}
As a corollary to the above, we arrive at the following characterization of equiangularity in matrix forms:
\begin{cor}\label{equiangular_matrix_form}
Assuming the same settings in \ref{equiangularity_PQP}, the matrices $P,Q$ after diagonalization of $Q$, are of the form:
\begin{align}\label{equiangular_matrix_forms}
Q=
\begin{pmatrix}
I_{d \times d} & 0 \\
0 & 0 
\end{pmatrix}, \ \ 
P=
\begin{pmatrix}
\alpha^2 I_{d\times d} & B \\
B^\dagger & D
\end{pmatrix},
\end{align}
where 
\begin{itemize}
    \item $BB^\dagger=(\alpha-\alpha^2)I_{d \times d},$
    \item $ B^\dagger B=\alpha^2D,$
    \item $D^2=(1-\alpha^2)D.$
\end{itemize}
In particular, $P$ and $Q$ have the same rank and if $\alpha=1$ then $P=Q$. Also, $P,Q$ are equiangular if they satisfy the matrix form and equations above.
\end{cor}
\begin{proof}
The first equation is derived from the fact that $U=\alpha^2 I_{d \times d}$ and equation (\ref{2-P=P^2}). The third equation is derived from the second and (\ref{2-P=P^2}). The second itself is derived by calculating $PQP$ in its matrix form and comparing its bottom right block $B^\dagger B$ to that of $\alpha^2P$ which is $\alpha^2D$. 

Take a vector $v=(v_Q,v_{I-Q})$ where the first coordinate is a $d$ dimensional vector and the rest lies in the kernel of $Q$. We want to find the dimension of the kernel of $P$. To have $Pv=0$, we must have $\alpha^2 v_Q=-Bv_{I-Q}, B^\dagger v_{Q}=-Dv_{I-Q}$. This reduces to one equation: $B^\dagger Bv_{I-Q}=\alpha^2 Dv_{I-Q}$ which is always true by the second equation in the statement. Hence the kernel has the same dimension as the kernel of $Q$.

Finally, it is easy to take the matrices in the statement and compute (\ref{r_Q_to_v},\ref{v_to_v_Q}) to show that indeed $||Pr_Q||,||Qr_P||$ is always a constant $\alpha$ as the equiangular definition requires.
\end{proof}

\begin{rmk}\label{equiangular_PQP_equivalence}
By analyzing the proof above, it can be seen that for two projections $P,Q$ satisfying $PQP=\alpha^2P,QPQ=\beta^2Q$, the projections are equiangular and $\alpha=\beta$ and the matrix form and equations in \ref{equiangular_matrix_form} hold for $P,Q$. Indeed, one can first diagonalize $Q$ and from $QPQ=\beta^2Q$ infer that $U=\beta^2 I_{d\times d}$ and get the rest of the results from $PQP=\alpha^2P$. Thus, equiangularity is also equivalent to $PQP=\alpha^2P,QPQ=\beta^2P$.
\end{rmk}

\begin{rmk}\label{No_intersection_remark}
An easy observable fact, yet very useful as we shall see in later sections, is that the equation $PQP=\alpha^2 P$ implies that $P$ and $Q$ do not intersect on a line unless they are equal. If not, there is $v$ such that $Pv=Qv=v$, which implies $\alpha=1$, i.e. $P=Q$. Hence, equiangularity in particular implies no nontrivial intersection of the planes involved.
\end{rmk}

Next, the notion of no-leakage pair of projections is described and proved to be equivalent to the notion equiangular pairs.

\begin{dfn}\label{no_leakage_dfn}
The pair of projections $P,Q$ are no-leakage if $PQP$ is a unitary map up to some scale on the image of $P$ and similarly for $QPQ$.
\end{dfn}
\begin{thm}
A no-leakage pair is the same as an equiangular pair.
\end{thm}
\begin{proof}
It is obvious that an equiangular pair is a no-leakage pair as the unitary map up to some scale is either $\alpha^2I_{d\times d}$ or $(1-\alpha^2)I_{d\times d}$. 

For a no-leakage pair, similar to previous theorems, we diagonalize $Q$ and see that $U$ must be a unitary up to some scale. Hence $UU^\dagger = \alpha^2 I_{d \times d}$ for some $\alpha>0$. But $U$ is hermitian and has only real eigenvalues, therefore only possible eigenvalues are $\pm \alpha$. Further, by (\ref{2-P=P^2}), we know that $U-U^2$ is positive, which means $U$ must be positive. Hence all eigenvalues of $U$ are $\alpha>0$ which implies that $U=\alpha I_{d \times d}$. 

Next, for the same condition on $PQP$, with the same argument above by exchanging the place of $P,Q$, we get $PQP=\beta^2 P$. By \ref{equiangular_PQP_equivalence}, we are done.
\end{proof}

\begin{rmk}
As pointed out by a referee, the discussion above can be understood in terms of invariant two dimensional subspaces. Indeed, two projections can always be simultaneously block-diagonalized to blocks of size at most $2\times 2$. This is a standard linear algebra fact. Take an eigenvector with positive eigenvalue of the nonnegative matrix $PQP$ such as $v$. Notice $Pv=v$. Then $v$ gives a pair $v,Qv$ which forms a subspace of dimension at most $2$ preserved by both $P$ and $Q$. We can discard this subspace and perform induction. If no positive eigenvalue exists, then $P$ and $Q$ are orthogonal. Hence, we can decompose the space into at most two dimensional subspaces which are mutually orthogonal, and invariant under both $P$ and $Q$. Equiangularity or no-leakage each hold if and only if the single angle that occurs in each block is the same for all blocks.
\end{rmk}

\section{Characterization of strong equiangularity}

A more restricted version of equiangular projections happens when not only $P,Q$ are equiangular but also $I-P,I-Q$ are equiangular. By \ref{equiangular_PQP_equivalence}, this is equivalent to $(I-Q)(I-P)(I-Q)=\gamma^2(I-Q),(I-P)(I-Q)(I-P)=\gamma^2(I-P)$.

\begin{dfn} \label{strong_equiangular}
Projections $P,Q$ are \textit{strongly} equiangular if they and their complements are equiangular.
\end{dfn}

By \ref{equiangular_matrix_form}, direct calculations result in $D=(1-\alpha^2)I$. More importantly, is the dimension of $D$. So far, we have only seen that $P,Q$ must have the same rank. Let $n$ be the dimension of the Hilbert space. We already had $BB^\dagger=(\alpha^2-\alpha^4) I_{d\times d}$. Hence, $B^\dagger$ can not have a kernel which means $d\le n-d$. But once $I-P,I-Q$ are equiangular, symmetrically, $B^\dagger B= (\alpha^2-\alpha^4) I$, which means $B$ cannot have a kernel, hence $d \ge n-d$. Therefore $B$ is a square matrix which is a unitary up to some scale and $d= \frac{n}{2}$. It turns out that this condition on $B$ along with the previous conditions in \ref{equiangular_matrix_form}, also imply that $P,Q$ are strongly equiangular. We summarize the findings into:

\begin{thm}\label{strongly_equiangular_thm}
$P,Q$ are strongly equiangular if and only if by a unitary transformation they become:
\begin{align}\label{strongly_equiangular_matrix_form}
Q=
\begin{pmatrix}
I_{d \times d} & 0 \\
0 & 0 
\end{pmatrix}, \ \ 
P=
\begin{pmatrix}
\alpha^2 I_{d\times d} & B \\
B^\dagger & (1-\alpha^2)I_{d\times d}
\end{pmatrix},
\end{align}
where $0<\alpha<1$, $B$ is a unitary up to scale $\alpha^2-\alpha^4$ and $P,Q$ have rank $d$, half of the Hilbert space dimension.
\end{thm}

\begin{rmk}\label{construction_thru_directsum_tensor}
By checking the requirement $PQP=\alpha^2P,QPQ=\alpha^2Q$, it can be easily shown that tensor product of equiangular pairs $(P_1,Q_1),(P_2,Q_2)$ is equiangular, and direct sum of equiangular is equiangular if and only if the scalar $\alpha$ of both pairs are equal. Further, only direct sum of strongly equiangular pairs is strongly equiangular while tensor product never is, as the rank can not be half of the Hilbert space dimension. 

Therefore, tensor product and direct sum are two constructions for obtaining equiangular projections in higher dimensions, where the former can be always used while the latter is more useful in discrete cases due to its restriction.
\end{rmk}

As we shall see, the complex, quaternions, and octonions will provide examples of a continuous family of mutually strongly equiangular projections.  Moreover, only octonions can be used for universal quantum computation.

For a strongly equiangular pair, each one is conjugate to the other by some unitary transformation as they have the same rank. We have the following lemma regarding this unitary conjugator.
\begin{lem}\label{lemma_U}
For a strongly equiangular pair $P,Q$ with rank $d$ as in (\ref{strongly_equiangular_matrix_form}), the unitary matrix that expresses $P$ in the orthogonal basis provided by $Q$ is 
\begin{align}
U=
\begin{pmatrix}
U_{11} & U_{12} \\
U_{21} & U_{22}
\end{pmatrix},
\end{align}
where $U_{ij}$ are $d\times d$ blocks and $U_{11},U_{22}$ are unitaries up to scale $\alpha$.

The opposite holds as well: If the unitary $U$ has the above property, then $Q,UQU^\dagger$ is a strongly equiangular pair for a diagonalized $Q$.
\end{lem}
\begin{proof}
It is simple linear algebra to see that for $P=UQU^\dagger$, with $P,Q$ in the form (\ref{strongly_equiangular_matrix_form}), one needs to have $U_{11}U_{11}^\dagger=\alpha^2 I_{d\times d}$ while $U_{12}U_{12}^\dagger=(1-\alpha^2)I_{d\times d}$. On the other hand $U$ itself must be unitary, which means that $UU^\dagger = I_{2d \times 2d}$. Computing the bottom right block gives $U_{22}U_{22}^\dagger= \alpha^2 I_{d \times d}$. 

By direct calculations, for unitary $U$ satisfying those properties, the projections $UQU^\dagger,Q$ are found to be of the form (\ref{strongly_equiangular_matrix_form}) for a diagonalized $Q$.
\end{proof}
Inspired by the lemma, we define a set:
\begin{align}
S=\{U \in U(2d)| U_{11},U_{22} \text{ are unitaries up to some scale } 0<\alpha<1 \}.
\end{align}
The above has the following consequence for collections of mutually strongly equiangular projections.
\begin{cor}\label{cor_S}
Consider a collection of projections $\{P_i\}_{i\in \mathcal{I}}$ with a distinguished diagonalized $P_0$ of rank $d$ and $P_i=U_i P_0 U_i^\dagger$ with $U_i$ unitaries and $U_0=I$. Then all pairs $\{P_i\}_{i\in \mathcal{I}}$ are strongly equiangular if and only if $U_i^\dagger U_j \in S, \ \forall i,j \in \mathcal{I}$. 
\end{cor}
\begin{proof}
As $P_i,P_j$ are strongly equiangular, the pair $U_i^\dagger P_i U_i, U_i^\dagger P_j U_i$ is also strongly equiangular. But $U_i^\dagger P_i U_i=P_0$ is diagonalized, hence \ref{lemma_U} applies, and as $U_i^\dagger P_jU_i = (U_i^\dagger U_j)P_0(U_i^\dagger U_j)^\dagger$, this implies $U_i^\dagger U_j \in S$. The converse holds using the converse in \ref{lemma_U}.
\end{proof}
Let $V^{i,j}=U_i^\dagger U_j$, a notation that will be used throughout the paper. 
\begin{rmk}\label{angle_P_1_P_2}
From \ref{lemma_U} and \ref{cor_S}, the dihedral angle between $P_i,P_j$ is given by simply calculating the $\arccos$ of the root of the scalar $V^{i,j}_{11}(V^{i,j}_{11})^\dagger$.
\end{rmk}
Finally, we would like to understand the unitary gates generated by a sequence of strongly equiangular projections. 
\begin{cor}\label{gates_cor}
With the same settings in \ref{cor_S}, satisfying mutually strong equiangularity, the operator $P_0P_{i_k}\ldots P_{i_1}P_0$ gives
\begin{align}
\begin{pmatrix}
\prod_{r=0}^k V^{i_{r+1},i_{r}}_{11} &  0_{d\times d} \\
0_{d \times d} & 0_{d\times d}
\end{pmatrix}
\end{align}
with $i_{k+1}=i_0=0$.
\end{cor}
\begin{proof}
Obvious, as $P_{i_r}=(U_{i_r}P_0) (P_0 U_{i_r}^\dagger)$ and $P_0V^{i_r,i_{r-1}}P_0=V^{i_r,i_{r-1}}_{11}$.
\end{proof}

\begin{rmk} \label{real_version}
Throughout this and the previous chapter, we assumed the Hilbert space to be over the field $\mathbb{C}$. But \textbf{all} results hold when $\mathbb{C}$ is replaced by $\mathbb{R}$, where we deal with real symmetric or orthogonal matrices which have real eigenvalues and real eigenvectors. Also (strongly) equiangular pairs in $\mathbb{R}^n$ are also (strongly) equiangular pairs in $\mathbb{C}^n$, where each coordinate is extended from real to complex which preserves entries of matrices hence preserving (\ref{equiangular_matrix_forms}) and (\ref{strongly_equiangular_matrix_form}).
\end{rmk}

Thus, by restricting ourselves to $P_0$, the unitary gate applied on $P_0$ (up to some scale) is the above product. In the next sections we ask what collections could give a universal quantum computer?

\section{Continuous family of strongly equiangular projections}
In this section, the possibility of division ring extensions of $\mathbb{C}$ is explored to provide a collection of strongly equiangular pairs. First, we recall the definition of octonions which contains all extensions of $\mathbb{C}$.

The non-associative division ring of octonions is generated by $e_i$, i.e. all elements of the octonions are of the form $o=\sum_{i=0}^7 o_ie_i$, with the multiplication Table \ref{tab1}.

\begin{table}
\caption{octonion multiplication table}
\label{tab1}
\begin{tabular}{|c|c|c|c|c|c|c|c|c|}
\hline $e_ie_j$ & $e_0$ & $e_1$ & $e_2$ & $e_3$ & $e_4$ & $e_5$ & $e_6$ & $e_7$ \\ \hline
$e_0$ & $e_0$ & $e_1$ & $e_2$ & $e_3$ & $e_4$ & $e_5$ & $e_6$ & $e_7$\\ \hline
$e_1$ & $e_1$ & $-e_0$ & $e_3$ & $-e_2$ & $e_5$ & $-e_4$ & $-e_7$ & $e_6$ \\\hline
$e_2$ & $e_2$ & $-e_3$ & $-e_0$ & $e_1$ & $e_6$ & $e_7$ & $-e_4$ & $-e_5$ \\\hline
$e_3$ & $e_3$ & $e_2$ & $-e_1$ & $-e_0$ & $e_7$ & $-e_6$ & $e_5$ & $-e_4$ \\\hline
$e_4$ & $e_4$ & $-e_5$ & $-e_6$ & $-e_7$ & $-e_0$ & $e_1$ & $e_2$ & $e_3$ \\\hline
$e_5$ & $e_5$ & $e_4$ & $-e_7$ & $e_6$  & $-e_1$ & $-e_0$ & $-e_3$ & $e_2$\\\hline
$e_6$ & $e_6$ & $e_7$ & $e_4$ & $-e_5$ & $-e_2$ & $e_3$ & $-e_0$ & $-e_1$ \\\hline
$e_7$ & $e_7$ & $-e_6$ & $e_5$ & $e_4$ & $-e_3$ & $-e_2$ & $e_1$ & $-e_0$ \\ \hline
\end{tabular}
\end{table}

The generators $e_0,e_1$ can be identified as the complex numbers $1,i$, while $e_0,e_1,e_2,e_3$ as $1,i,j,k$ to form the quaternions. Most of the times the generator $e_0$ will simply be replaced by $1$. The product rule can be written as:
\begin{align}
e_i.e_j=\begin{cases}
e_i, & \text{if  } j=0\\
e_j, & \text{if } i=0\\
-\delta_{ij}+\varepsilon_{ijk}e_k, & \text{otherwise},
\end{cases}
\end{align}
where $\varepsilon_{ijk}$ is a completely anti-symmetric tensor with $+1$ value only for $ijk=123,145,176,246,257,347,365$ and their cyclic permutations.

Elements of quaternions (and octonions) can also be represented as a sum of $2$ (or $4$) complex numbers by $a=z_0+z_1e_2$ (or $a=\sum_{i=0}^3 z_ie_{2i}$).  Conjugation in each ring is defined as the flipping of signs in all $e_i, i>0$. Therefore, $o^*=o_0e_0 - \sum_{i=1}^7 o_ie_i$. Further in all the three rings $aa^*=a^* a=||a||^2$, where $||a||^2$ is the sum of squared of the real numbers representing $a$. This is very similar to a unitary up to a scale, which is the property that turns out to be important.~
\\
\subsection{Four families of equiangular planes}~
\\
Using these division rings to generate strongly equiangular pairs comes from the geometric picture. There exist exactly four $n$-dimensional equiangular families of $n$-planes in $\mathbb{R}^{2n}$ for $n\in \{1,2,4,8\}$.

The space of lines in $\mathbb{C}^2$ form $\mathbb{C}P^1 \cong S^2$, and they obviously have the strongly equiangular property as they are simply one dimensional lines. The unitary gates on $\mathbb{C}$ that these projections provide consist of all complex numbers. Alternatively, viewed as strongly equiangular planes in $\mathbb{R}^4$, the gates form $SO(2)$. Notice that the gates are not strictly unitary but always unitary up to some scale. Therefore, precisely an isomorphic copy of $\mathbb{R} \times SO(2) \cong \mathbb{C}$ is recovered. But the scaling will be ignored due to normalization, so we are talking about complex numbers of norm one or equivalently $SO(2)$.

Taking this argument one step further, one could consider the \textit{quaternionic lines} forming the space $\mathbb{H}P^1 \cong S^4$. As they are \textit{lines}, they should also form strongly equiangular pairs. Through a suitable embedding, one expects them to form strongly equiangular pairs of $4$-dimensional real planes in $\mathbb{C}^4 \cong \mathbb{R}^8$. As for the gates, since quaternions themselves have a unitary up to some scale representation in $M_{\mathbb{C}}(2,2)$, going from one quaternion to another (from an $\mathbb{R}^4$ subspace to another) is an action by a quaternion (a unitary up to scale). Thus we can expect to recover all gates in $SU(2)$.

The last generalization is to octonions $\mathbb{O}$ and consider octonionic lines forming $\mathbb{O}P^1 \cong S^8$. They are $8$-real dimensional planes in $\mathbb{C}^8 \cong \mathbb{R}^{16}$. One might expect to recover $SU(4)$, but we recover something more: $SO(8)$ which contains $SU(4)$. Notice while $SO(2)$ and $SU(2)$ have real dimensions $1,3$, $SO(8)$ has real dimension $28$. The reason for this jump in dimension is the non-associativity of the octonions which means there is no representation of the octonions as linear operators on any vector space, in particular $\mathbb{C}^4$ (or $\mathbb{R}^8$). What we will do is to choose some ``representation'' of the octonions as linear operators, but this map will not be a homomorphism of algebras. Yet, it turns out to provide a universal quantum gate set. The case of octonions is done separately in section \ref{section_Octonions}.

Are there any more examples of continuously parametrized strongly equiangular projections? The above pattern suggests that such collections of projections could be related to division ring extensions of $\mathbb{R}$, of which there are only three. Recall Remark \ref{No_intersection_remark}, where it was proven that a pair of equiangular projections only intersect at one point, the origin. 

Now assume $\mathcal{P}$ is a submanifold of dimension $n$ in the Grassmannian $Gr(2n,n)$, with the property that $n$-dimensional planes corresponding to the points in $\mathcal{P}$ have no nontrivial intersection. The intersection of each plane $p \in \mathcal{P}$ with the unit sphere $S^{2n-1}$ creates a sphere $S^{n-1}$, and all such $(n-1)$-spheres are disjoint, with linking number one, as the assumption implies. Since $\mathcal{P}$ has dimension $n$, there is a \textit{local} fibration of $S^{2n-1}$ with base a local chart of $\mathcal{P}$, isomorphic to an open $n$-disk $D^n$.
\begin{thm}
Only for $n \in \{1,2,4,8\}$ can there exist an $n$-parameter family of embedded $S^{n-1}$'s in $S^{2n-1}$, each pair with linking number one.
\end{thm}
\begin{proof}
Counting dimensions the image of the germ of the given family constitutes a region $X \subset S^{2n-1}$ which is an $S^{n-1}$ bundle over an open disk $D^n$. We may construct $f: S^{2n-1} \to S^n$ by mapping all of $S^{2n-1}\backslash X$ to the south pole $s$ of $S^n$ and then projecting out the fibers of $X$ to $D^n$ followed by the degree one mapping of $D^n$ to $S^n\backslash \{s\}$. The Hopf invariant of $f$ may be computed as the linking number of generic point preimages; which clearly is $1$. By Adam's Theorem, Hopf invariant one only occurs for the Hopf maps \cite{adams1960non} which only exist for $n\in \{1,2,4,8\}$.
\end{proof}
As the smooth structure of $\mathcal{P}$ was never used in the above arguments, the following holds:
\begin{cor}\label{R^2n_classification_equiangular}
Given a \textit{continuous} $n$-dimensional manifold of equiangular $n$ planes in $\mathbb{R}^{2n}$, we have $n \in \{1,2,4,8\}$.
\end{cor}

The above suggests:
\begin{ques}
If there is a family (discrete or continuous) of subspaces with no zero angles, is the family deformable to an equiangular one?
\end{ques}
One idea in this direction is to write down the parabolic (heat) equation which follows the gradient towards equiangularity. We leave this to the future.

We can prove a stronger theorem, where the parameter space is smaller:
\begin{thm}
Only for $n \in \{1,2,4,8\}$ can there exist an $n-1$ or $n-2$ (if $n>3$) parameter family of embedded $S^{n-1}$'s in $S^{2n-1}$, each pair with linking number one.
\end{thm}
\begin{proof}
We try to extend the $k$ family to an $n$-family and apply the previous theorem. Denoting the image of the $k$-parameter family by $X$, we note that it is the oriented product manifold $D^k\times S^{n-1}$. If its (oriented) normal bundle $\nu(X)$ (with fiber dimension $j=1,2$) inside $S^{2n-1}$ can be parallelized, one can simply extend $X$ in the trivial way to a local region like the one in the previous theorem; note that the linking number, using an easy continuity argument, will remain one. The group of oriented bundles over $S^{n-1}$ with fiber $\mathbb{R}^j$ is well-known to be in bijection with $\pi_{n-2}(SO(j))$. Note this classification is up to homotopy type, and $S^{n-1}\times D^k,S^{n-1}$ have the same homotopy type. For $j=1,2$ corresponding to $k=n-1,n-2$, we know $\pi_{n-2}(SO(j))$ is trivial ($n>3$ if $j=2$) as $SO(1)=\{\text{point}\},SO(2)=S^1$.
\end{proof}

\begin{rmk}\label{S^3_hom}
We cannot extend the above theorem to $k=n-3,n-4$, at least not with the above argument as it is known that higher homotopy groups $\pi_{n-2}(S^3)$ for $n>4$ are non-vanishing. More precisely, for $k=n-3$, we know $SO(j=3)\cong S^3$, and for $k=n-4$, we know $SO(j=4)\cong S^3 \times S^3$, which homotopy groups are a product of those of $S^3$. So for the above argument to apply to any of these two cases, we need a vanishing higher homotopy group for $S^3$.
\end{rmk}

\begin{rmk}\label{SO(n-k)_hom}
Using the tables in \cite[p.258-260]{lundell1992concise} for $\pi_{n-2}(SO(n-k))$, it is observed that there is one more family of pairs $(n,k)$ for which $\pi_{n-2}(SO(n-k))$ is vanishing:
$$\pi_{n-2}(SO(n-k))=0 \text{ for } n\equiv 6\pmod 8, n\ge 22, k = 5 $$
This implies the absence of a smooth $k-$dimensional family of strong equiangular $n$-planes for such pairs.
\end{rmk}

\begin{rmk}
By \textit{Remark} \ref{construction_thru_directsum_tensor}, One can use the direct sum on these families (by taking $P_a \oplus P_a$, where $P_a$ is defined in the next section) to produce strongly equiangular families with number of parameters $k=\min\{2^{m},8\}$ for $n$ divisible by $2^m$.

One can also start with the $2n=2,4,8,16$ family and tensor it with $\mathbb{C}$ to get a real equiangular $4,8,16,32$ example. This is of course not \textit{strongly} equiangular as the plane dimensions are a quarter of the Hilbert space dimension. Then forgetting the complex structure and tensoring it again with $\mathbb{C}$, and repeating it, one can get different examples of continuous equiangular collections in $\mathbb{R}^{k}$ for $k=2^{m-2},2^{m-1},2^m, 2^{m+1}$ for all $m>2$. 
\end{rmk}

Therefore, we have the plot in Figure \ref{strongly_equiangular_plot} on the smooth strong equiangular families for $n,k\le 16$. In addition to the above theorems and remarks, we note the obvious fact that any sub-family of the ones found, also form a strongly equiangular family. This means a ``\textbullet'' at a point $(n,k)$ implies one at all points $(n,k')$ for $k'\le k$. Conversely, the known absence of a strong equiangular family denoted by ``$\circ$'' at $(n,k)$, implies the same for $(n,k'), k'\ge k$.
\begin{figure}[h]
\centering
\begin{tikzpicture}[x=0.6cm,y=0.6cm]
\draw[-latex, thin, draw=gray] (0,0)--(17,0) node [right] {$n$}; 
\draw[-latex, thin, draw=gray] (0,0)--(0,17) node [above] {$k$}; 
\draw [step=1.0,dotted, gray] (0,0) grid (16,16);
\foreach \n in {0,...,16} {
        \node at (\n,0) [below] {\tiny $\n$};
        \node at (0,\n) [left] {\tiny $\n$};
}
\foreach \n in {3,5,6,7,9,10,11,12,13,14,15,16} {
        \node at (\n,\n)  {$\circ$};
        \node at (\n,\n -1) {$\circ$};
}
\foreach \n in {5,6,7,9,10,11,12,13,14,15,16} {
        \node at (\n,\n -2) {$\circ$};
}
\foreach \n in {1,...,16} {
        \node at (\n,1)  {\textbullet};
}
\foreach \n in {2,4,6,8,10,12,14,16} {
        \node at (\n,2)  {\textbullet};
}
\foreach \n in {4,8,12,16} {
        \node at (\n,4)  {\textbullet};
}
\foreach \n in {8,16} {
        \node at (\n,8)  {\textbullet};
}
\foreach \n in {1,2,4,8} {
        \node at (\n,1)  {\textbullet};
}
\foreach \n in {4,8,12,16} {
        \node at (\n,3)  {\textbullet};
}
\foreach \k in {5,6,7} {
        \node at (8,\k)  {\textbullet};
        \node at (16,\k)  {\textbullet};
}
\end{tikzpicture}
\caption{Smooth $k$-dimensional family of strongly equiangular $n$-planes in $\mathbb{R}^{2n}$. ``\textbullet'' represents the known strongly equiangular families, and ``$\circ$'' represents the known absence of such a family. The case of non-annotated points is currently unknown except. For example, we do not know if there is a $2$-parameter family in $10$-space, i.e. $(5,2)$ above. Of course, where $k>n$, dimensions add up to more than $2n$, and the families cannot exist by invariance of domain.}\label{strongly_equiangular_plot} 
\end{figure}
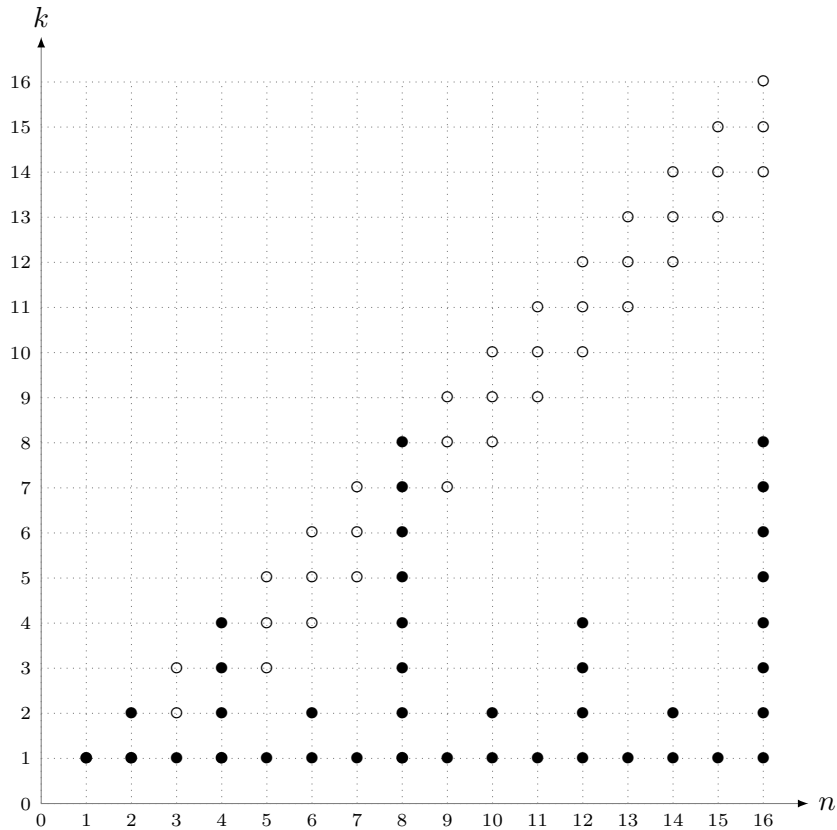

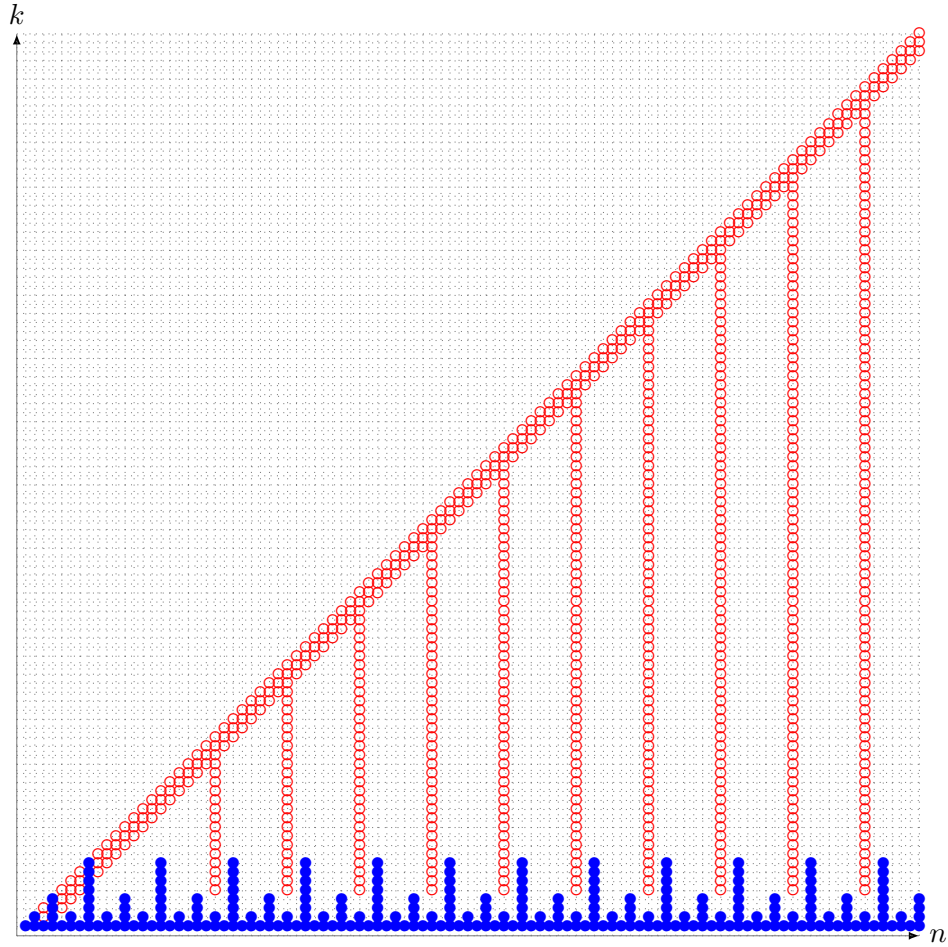
\begin{figure}[h]
\centering
\begin{tikzpicture}[x=0.12cm,y=0.12cm]
\draw[-latex, thin, draw=gray] (0,0)--(100,0) node [right] {$n$}; 
\draw[-latex, thin, draw=gray] (0,0)--(0,100) node [above] {$k$}; 
\draw [step=1.0,dotted, gray] (0,0) grid (100,100);

\foreach \n in {3,5,6,7,9,10,11,12,13,14,15,16} {
        \node at (\n,\n)  {$\red{\circ}$};
        \node at (\n,\n -1) {$\red{\circ}$};
}
\foreach \n in {5,6,7,9,10,11,12,13,14,15,16} {
        \node at (\n,\n -2) {$\red{\circ}$};
}
\foreach \n in {17,...,100} {
        \node at (\n,\n)  {$\red{\circ}$};
        \node at (\n,\n -1) {$\red{\circ}$};
}
\foreach \n in {17,...,100} {
        \node at (\n,\n -2) {$\red{\circ}$};
}

\foreach \n in {1,...,100} {
        \node at (\n,1)  {\blue{\textbullet}};
}
\foreach \n in {1,...,50} {
        \node at (2*\n,2)  {\blue{\textbullet}};
}
\foreach \n in {1,...,25} {
        \node at (4*\n,4)  {\blue{\textbullet}};
        \node at (4*\n,3)  {\blue{\textbullet}};
}
\foreach \n in {1,...,12} {
        \node at (8*\n,8)  {\blue{\textbullet}};
        \node at (8*\n,7)  {\blue{\textbullet}};
        \node at (8*\n,6)  {\blue{\textbullet}};
        \node at (8*\n,5)  {\blue{\textbullet}};
}

\foreach \n in {2,...,11} {
    \pgfmathsetmacro{\z}{8*\n+3}
    \foreach \k in {5,...,\z} {
        \node at (8*\n+6,\k)  {$\red{\circ}$};
        }
}
\end{tikzpicture}
\caption{A color-coded plot demonstrating the pattern of smooth $k-$dimensional family of strongly equiangular $n-$planes in $\mathbb{R}^{2n}$ for $n,k \le 100$. The color red is for the known absence. Notice the particular family in Remark \ref{SO(n-k)_hom} starts at $(n=22,k=5)$ with the general rule $n\equiv 6 \pmod 8, n\ge 22, k\ge 5$.}\label{larger_plot}
\end{figure}
~
\\
\subsection{Strongly equiangular pairs from lines}\label{High_dimensional_lines_and_strongly_equiangular_pairs}~
\\
A quaternionic or octonionic line means vectors of form $\binom{x}{Ax}$, where $A$ is actually a matrix representation of an element $a$ inside the ring $\mathcal{A}$. Our division rings are the extensions of the complex numbers, and the line is the graph of the linear function $y=Ax$. What is meant by a representation is not necessarily an algebra homomorphism, although it will have some naturality, and be a homomorphism for complex numbers and quaternions.

Similar to the previous sections, results in this part also hold when the vector space is over the real numbers. First, one needs to derive the projections $P_a$ on these linear spaces. It is not hard to see that:
\begin{align}\label{P_a}
P_a = 
\begin{pmatrix}
T_A^{-1}  & T_A^{-1}A^\dagger  \\
AT_A^{-1}  & AT_A^{-1} A^\dagger
\end{pmatrix}
\end{align}
where $T_A=(I+A^\dagger A)$. Any representation of the extension rings will have the element $A=0$, and 
\begin{align}
P_0 = 
\begin{pmatrix}
I_{d\times d}  & 0_{d\times d} \\
0_{d\times d} & 0_{d\times d}
\end{pmatrix}
\end{align}
is already diagonalized. For $P_a$ and $P_0$ to be strongly equiangular, it is necessary (and sufficient) that $T_A=\alpha^{-2} I_{d\times d}$ which is equivalent to $A$ being a unitary up to some scale. This is not an unwelcome restriction as long as this \textit{same} property of the elements inside the extension rings is preserved: their inverse is a scalar multiple of their conjugate, i.e. $aa^* = ||a||^2$. 

Therefore, we assume that there is a nice representation with $A \neq 0$ always some unitary up to a scale. This makes $P_a$ of the form
\begin{align}\label{P_a_unitary}
P_a = 
\begin{pmatrix}
\frac{1}{1+||A||^2}I_{d\times d}  & \frac{1}{1+||A||^2}A^\dagger  \\
\frac{1}{1+||A||^2} A  & \frac{||A||^2}{1+||A||^2} I_{d\times d}
\end{pmatrix}
\end{align}
where $AA^\dagger=||A||^2I_{d\times d}$. The next step is to understand when $P_a,P_b$ for $a,b\in \mathcal{A}$ can be strongly equiangular.

The unitary $U_a$ diagonalizing $P_a=U_aP_0U_a^\dagger$ can be computed directly:
\begin{align}\label{U_a}
U_a=
\begin{pmatrix}
        \frac{1}{\sqrt{1+||A||^2}}I_{d\times d} &  \frac{-1}{\sqrt{1+||A||^2}}A^\dagger \\
        \frac{1}{\sqrt{1+||A||^2}}A & \frac{1}{\sqrt{1+||A||^2}}I_{d\times d}
\end{pmatrix}.
\end{align}

\begin{lem}\label{1+ab_lem}
Linear spaces $\binom{x}{Ax}$ and $\binom{x}{Bx}$ for two matrices $A,B \in \mbbC^{d\times d}$, which are unitaries up to some scale, form a strongly equiangular pair if and only if 
\begin{align}
I_{d\times d}+A^\dagger B
\end{align}
is a unitary matrix up to some scale.
\end{lem}
\begin{proof}
This is straightforward application of \ref{cor_S}, where for $V^{a,b}=U_a^\dagger U_b$, the diagonal blocks are
\begin{align}\label{diagonal_blocks}
V^{a,b}_{11}=\frac{I_{d \times d} + A^\dagger B}{(1+||A||^2)^{\frac{1}{2}}(1+||B||^2)^{\frac{1}{2}}} \\ 
V^{a,b}_{22}=\frac{I_{d\times d} +B^\dagger A}{(1+||A||^2)^{\frac{1}{2}}(1+||B||^2)^{\frac{1}{2}}}.
\end{align}
\end{proof}
\begin{rmk}
Although we will only construct special cases of a collection of strongly equiangular projections, one can ask whether there is a classification for such collections. For example, $I_{d\times d}+A^\dagger B$ is a unitary up to some scale if and only if $A^\dagger B$ is a unitary matrix with two eigenvalues complex conjugate of another, unless $A^\dagger B \propto I_{d\times d}$, which is a significant restriction on the choices of these matrices.
\end{rmk}
~
\\
\subsection{Gates from complex numbers and quaternions}~
\\
Hilbert spaces in this section are over the complex field. Assume $d=1$. Then $A \in \mathbb{C}$, and \ref{1+ab_lem} obviously holds. This is equivalent to considering the obvious representation for $\mathcal{A}=\mathbb{C}$ as linear maps on $\mathbb{C}$.

For $d=2$, notice that the set of unitaries up to some scale in $M_\mathbb{C}(2,2)$ is isomorphic to the \textbf{algebra} of quaternions. This means that the set of unitaries up to a scale in $d=2$ is closed under multiplication and addition, making $I_{d\times d}+A^\dagger B$ also a unitary up to some scale. Therefore, the maximal collection of projections $P_a$ in $d=2$ is a collection of strongly equiangular pairs. Again, like the previous case, we see the usual representation of quaternions $a=(a_0+ja_2)+i(a_1+ja_3)$ as 
\begin{align}\label{unitary_2d}
A=
\begin{pmatrix}
z_0 & z_1 \\
-\overline{z_1} & \overline{z_0}
\end{pmatrix}, \ \ z_0=a_0+ia_2, \ \ z_1=a_1+ia_3,
\end{align}
easily works. But do we get universal gates? 
\begin{thm}
Unitary gates from sequence of projections given by representations of $\mathbb{C},\mathbb{H}$ generate $SO(2),SU(2)$ respectively.
\end{thm}
\begin{proof}
Using the notations in \ref{gates_cor}, 
$$V^{ij}_{11}=\frac{1+A_i^\dagger A_j}{(1+||A_i||^2)^\frac{1}{2}(1+||A_j||^2)^\frac{1}{2}}.$$ 
Recall the operator $P_0P_{i_k}\ldots P_{i_1}P_0$ gives
\begin{align}
\begin{pmatrix}
\prod_{r=0}^k V^{i_{r+1},i_{r}}_{11} &  0_{d\times d} \\
0_{d \times d} & 0_{d\times d}
\end{pmatrix}
\end{align}
with $i_{k+1}=i_0=0$. The operator on the states inside $P_0$ is
\begin{align}\label{gates_formula}
\frac{\prod_{r=0}^{k}(1+A_{i_{r+1}}^\dagger A_{i_r})}{\prod_{r=0}^{k}(1+||A_{i_{r}}||^2)}.
\end{align}

When $A_i \in \mathbb{C}$, then $\prod_{r=0}^k V^{i_{r+1},i_{r}}_{11}$ falls into $\mathbb{C}$ as well. Further, taking $P_0P_1P_bP_0$, gives $\frac{1+B}{(1+||B||^2)(1+||1||^2)}$ where actually $B=b$. Hence, by normalization and varying $b$, complex numbers of norm one or equivalently $SO(2)$ is recovered.

Similarly for the case of quaternions, as they form an algebra, the product $\prod_{r=0}^k V^{i_{r+1},i_{r}}_{11}$ is in the algebra, which means after normalization they can not give anything more than $SU(2)$. Choosing $P_0P_1P_bP_0$ gives the unitary up to scale $\frac{1+B}{(1+||B||^2)(1+||1||^2)}$. By normalization and varying $B$, $SU(2)$ is recovered.
\end{proof}
We can go over what was done over real vector spaces. The real representation would actually be more in line with what we desire to represent. Notice we would like to see an embedding of $\binom{x}{ax}$, where $ax$ is supposed to be a multiplication inside the ring. After all, the inspiration was projective lines, which are exactly of that form. We need a vector representation of the ring at the same time as an operator representation and would like the vector $Ax$ to be representative of the element $ax$. This uniquely defines $A$ as the vector representation is an isomorphism of vector spaces, hence defining $A$ over a basis. Although, from a practical perspective, by switching to the real representation, the gates that will be recovered will not change as a result.

For $d=2$, by choosing the obvious vector representation $x=x_1+ix_2 \in \mathbb{C} \to \binom{x_1}{x_2} \in \mathbb{R}^2$, the matrix representation for $a=a_0+ia_1$ is as follows:
\begin{align}\label{C_real_rep}
\begin{pmatrix}
a_0 & -a_1 \\
a_1 & a_0
\end{pmatrix}.
\end{align}
For $d=4$ and the quaternions, the vector representation for $x=\sum_{i=0}^7 x_ie_i$ is similarly $(x_0,x_1,x_2,x_3)$. Then, the matrix representation for $a=\sum_{i=0}^3 a_i e_i$ is:
\begin{align}
A=
\begin{pmatrix}
a_0 & -a_1 & -a_2 & -a_3  \\
a_1 & a_0 & -a_3 & a_2  \\
a_2 & a_3 & a_0 & -a_1 \\
a_3 & -a_2 & a_1 & a_0
\end{pmatrix}.
\end{align}
Note this is an extension of the real representation of complex numbers in (\ref{C_real_rep}). This representation not only works on the matrix-vector level to represent the product, but even as composition of matrices, it is actually a homomorphism of the quaternions to $4 \times 4$ real matrices. This was similarly true for the case of complex numbers. This will not be true for octonions as they are non-associative.

The gates recovered will not change, as in fact, the above matrix can be seen to be of the form
\begin{align}\label{why_SU(2)}
A=
\begin{pmatrix}
A_1 & -A_2\\
A_2 & A_1
\end{pmatrix}
\end{align}
where $A_1+iA_2$ gives the matrix in (\ref{unitary_2d}). That is why quaternions do not give $SO(4)$ which might be expected as $\mathbb{C},\mathbb{O}$ give $SO(2),SO(8)$. 

In the case for octonions, the last representation will be extended. 

\section{Universal computing model from octonions}\label{section_Octonions}
\subsection{$SO(8)$ from octonions}~
\\
We shall work with $\mathbb{R}^8$ and $\mathbb{R}^{16}$. The ``representation'' for octonion multiplication by $a=\sum_{i=0}^7 a_ie_i$ on $x=(x_0,\ldots,x_7)$ is
\begin{align}\label{octonion_matrix}
A=\begin{pmatrix}
    a_0 & -a_1 & -a_2 & -a_3 & -a_4 & -a_5 & -a_6 & -a_7\\
    a_1 & a_0 & -a_3 & a_2 & -a_5 & a_4 & a_7 & -a_6 \\
    a_2 & a_3 & a_0 & -a_1 & -a_6 & -a_7 & a_4 & a_5 \\
    a_3 & -a_2 & a_1 & a_0 & -a_7 & a_6 & -a_5 & a_4 \\
    a_4 & a_5 & a_6 & a_7 & a_0 & -a_1 & -a_2 & -a_3 \\
    a_5 & -a_4 & a_7 & -a_6  & a_1 & a_0 & a_3 & -a_2\\
    a_6 & -a_7 & -a_4 & a_5 & a_2 & -a_3 & a_0 & a_1 \\
    a_7 & a_6 & -a_5 & -a_4 & a_3 & a_2 & -a_1 & a_0
\end{pmatrix}.
\end{align}
The matrix is the same $8\times 8$ matrix in Table \ref{tab1} with the following modifications: all columns except the first one are multiplied by $-1$. This matrix is no longer of the form (\ref{why_SU(2)}), which is why more than $SU(4)$ is recovered. Also, as mentioned before, the matrix representation is not a homomorphism to matrix algebras as the octonions are non-associative.

In order to apply results in \ref{High_dimensional_lines_and_strongly_equiangular_pairs}, the first requirement is to show that $AA^\dagger \propto I$. 

Observe that $A=a_0I+H_a$ where $H_a$ is a skew-symmetric matrix and from the definition of $A$, $A^*$ which corresponds to $a^*$ is indeed equal to $A^\dagger$ as $H_{a^*}=-H_a=H_a^\dagger$. Therefore $AA^\dagger=$
\begin{align}
(a_0I+H_a)(a_0I+H_{a^*})=a_0^2 I + a_0(H_a+H_{a^*}) +H_aH_{a^*}.
\end{align}
The second term is zero as $-H_a=H_{a^*}$. The third term by direct calculation turns out to be $(\sum_{i=1}^7 a_i^2)I$. Hence, $AA^\dagger=||a||^2I$. This also implies that the complement to $P_a$, the projection on $\binom{x}{Ax}$, is $P_{-a^*/||a||^2}$, which was also the case for the complex and quaternionic lines.

The remaining step is to check whether $I+A^\dagger B$ is a unitary up to some scale. By simple calculations, this means
\begin{align}\label{ab+ba}
(I+A^\dagger B)(I+B^\dagger A) \propto I \leftrightarrow
A^\dagger B + B^\dagger A \propto I \leftrightarrow 
\end{align}
$$
(a_0I-H_a)(b_0I+H_b)+(b_0I-H_b)(a_0I+H_a) \propto I $$
The last term is equal to $2(\sum_{i=0}^7 a_ib_i)I$. This is also equal to twice the inner product of $a,b$ as $8$-dimensional vectors. 

Summarizing, the octonions give a collection of strongly equiangular pairs. We want to prove the gates they generate is $SO(8)$. The idea will be to show that the $8$ matrices corresponding to $e_i$ called $E_i$ will be enough to generate the Lie algebra $\mathfrak{so}(8)$ ($28$ dimension) by $E_i^\dagger E_j, i \neq j$. Then the formula $I+A_i^\dagger A_{i+1}$ appearing in gates formula (\ref{gates_formula}) will be related to the exponential of an element inside the Lie algebra. As the Lie algebra is generated by those elements, the Lie group will be generated by their exponentials. 

We now fill in the details of the above idea:
\begin{thm}\label{E_i_generate_so(8)}
Let $E_i$ be the matrix corresponding to the octonion $e_i$. Then $E_i^\dagger E_j$ for $i\neq j$, gives a basis for the Lie algebra $\mathfrak{so}(8)$.
\end{thm}
\begin{proof}
As $E_i^\dagger= \pm E_i$, we need to prove the theorem for $E_iE_j$. First we analyze what $E_i$ is. Entries are indexed by rows and columns numbered $0$ to $7$ top to bottom and left to right.

Each entry in $(s,l)$ of $E_i$ is given by $c_{(s,l)}^i \in \{0,\pm1\}$ where $c_{(s,l)}^ie_ie_l=e_s$. In other words $-c_{(s,l)}^ie_i=e_se_l$ for $l\neq 0$. So the entry $(s,l)$ for $l \neq 0$, is $c_{(s,l)}^i=-\varepsilon_{sli}=\varepsilon_{lsi}$. For $l=0$, $s=i$ and $c_{(s,l)}^i=1$. Each row and column have only one nonzero entry.

Assume $i,j\neq 0$. Notice $E_iE_j$ is a skew-symmetric matrix as $E_i= H_{e_i} \implies (E_iE_j)^\dagger=E_jE_i$ and using (\ref{ab+ba}), one obtains $H_aH_b+H_bH_a=-2(\sum_{i=1}^7a_ib_i)I$
which is zero when $a=e_i,b=e_j, i \neq j$. Thus, only the upper diagonal entries of $E_iE_j$ needs to be computed.

The entry in $(s,m)$ of $E_i E_j$ is given by the dot product of the row $s$ and column $m$ in $E_i$ and $E_j$, respectively. As the row $s$ and column $m$ have only one nonzero entry, there must exist an $e_l$ such that $c_{(s,l)}^ie_ie_l=e_s$ and $c_{(l,m)}^je_je_m=e_l$. There are three (actually two) special cases:
\begin{itemize}
    \item If $l=0$, then $s=i,m=j$ implying $c_{(s,l)}^i=1,c_{(l,m)}^j=-1$ and the result is $-1$ for the entry $(i,j)$.
    \item If $m=0$, then $l=j$, and $s=k$ where $e_ie_j=\varepsilon_{ijk}e_k$ and the result in $(k,0)$ is $\varepsilon_{ijk}$. 
    \item If $s=0$, then $l=i$ and $m=k$ as above and the entry in $(0,k)$ is $-\varepsilon_{ijk}$ due to skew-symmetry. 
\end{itemize}

Otherwise, we have $c_{(s,l)}^ic_{(l,m)}^j=\varepsilon_{sli}\varepsilon_{lmj}$. But rearranging the equations,  $(c_{(s,l)}^ie_ie_l) (c_{(l,m)}^je_le_j)=e_se_m$. Notice due to non-associativity the parentheses need to be preserved. But through direct calculations, it is shown that a property of the octonions is $(e_ie_l) (e_le_j) \propto e_ie_j, \forall i,j,l$. Therefore $e_se_m \propto e_ie_j \propto e_k$. 

This implies that the matrix $E_iE_j$ has nonzero entries exactly where $E_k$ does. It is only the signs of the $\pm 1$ entries that are changed. Through direct computations using the anti-symmetric tensor $\varepsilon$, it can be checked that always one of the four $\pm 1$ entries in the upper diagonal part of the matrix has a sign different from the three others. Let us call that entry the \textit{distinguished} entry. 

For any $i$ and a fixed $k$, there is a unique $j$ such that $e_ie_j \propto e_k$. Thus there are four pairs for each fixed $k$ of such $(i,j)$ and each pair gives a matrix $E_iE_j$. With the help of computer, each of these matrices can be seen to have a unique and different distinguished entry from the others.

For proving all $E_iE_j$s are linearly independent, the linear independency needs to be checked for each group of four matrices which have matching locations of the nonzero entries, and these matrices as explained above are linearly independent as:
\begin{align}
\text{det}
\begin{pmatrix} 
-1 & 1 & 1 & 1 \\
1 & -1 & 1 & 1 \\
1 & 1 & -1 & 1 \\
1 & 1 & 1 & -1 
\end{pmatrix}
\neq 0.
\end{align}
The $28=\binom{8}{2}$ total choices of $E_iE_j$ are linearly independent skew-symmetric matrices, thus spanning $\mathfrak{so}(8)$.
\end{proof}

The desired theorem for obtaining universal gates is:
\begin{thm}\label{octonion_universality}
The operators $P_{0}P_{a_k}\ldots P_{a_1}P_{0}$, where $a_j=t_je_{k_j}$ for some $t_j \in \mathbb{R}$ and index $k_j \in \{0,\ldots,7\}$, once normalized, generate the Lie group $SO(8)$.
\end{thm}
\begin{proof}
It is a standard theorem that for a basis $\{h_i\}$ of the Lie algebra of a finite dimensional compact Lie group, all elements of the Lie group are of the form $\prod_i \text{exp}(t_ih_i)$ for a suitable choice of $t_i \in \mathbb{R}$. It was also shown that $P_{0}P_{a_k}\ldots P_{a_1}P_{0}$ gives the unitary gate $C_{a_k,\ldots,a_1} \prod_i (1+A_{i+1}^\dagger A_i)$ where $C_{a_k,\ldots,a_1} \in \mathbb{R}$ is some scalar to make the product unitary.

The previous theorem gives the basis $\{h_i\}=\{E_iE_j\}$, but we also need $\exp(t_ih_i)$ to be of the form $(1+d_ih_i)$ for some real number $d_i$. This is indeed the case as the basis is not only a collection of skew-symmetric matrices but also matrices which are unitary up to a scale. Hence using the standard Taylor expansion of $\exp()$ of a matrix, we have the desired form for $\exp(t_ih_i)$, up to some scale which will factor into the scalar $C_{a_k,\ldots,a_1}$.

The remaining subtlety is that the choice of $a_i$ effects three gates. This means, e.g. choosing $a_3=e_1$ implies that the next gate defined by $h_3=E_{3}^\dagger E_{4}$ has only $8$ possibilities. In order to solve this issue, consider sequences of the form 
\begin{align}\label{Gates}
P_0P_{a_{2k}}P_{a_{2k-1}}P_0 P_0 P_{a_{2k-2}}P_{a_{2k-3}} P_0 P_0 \ldots P_0 P_{a_2}P_{a_1} P_0.
\end{align}
Then the gate is $C_{a_{2k},\ldots,a_1}\prod_i (1+A_{2i}^\dagger A_{2i-1})$ allowing to exactly apply any combination of $\prod_l \text{exp}(t_lE_{k_l}E_{j_l})$ by choosing $a_{2l}=e_{j_l},a_{2l-1}=c_le_{k_l}$ where $c_l$ is determined such that $(1\pm c_lE_{j_l}E_{k_l}) \propto \text{exp}(t_lE_{k_l}E_{j_l})$ ($\pm$ is dependent on whether $j_l=0$ or not). In fact, as $(E_{k_l}E_{j_l})^2=-I$,  $c_l$ is $\tan(t_l)$.
\end{proof}

It seems enough gates are there to have a universal quantum computational model. But the details of how forced measurements would work have not been discussed yet. First, one needs to show that the measurements, even if they fail at times, can in the end succeed to implement efficiently a desired local gate.~
\\
\subsection{Efficient implementation of local gates by forced measurements}\label{subsection:efficient_implementation}~
\\
Assume one wants to apply a unitary gate in $SO(8)$ on an $8-$qubit $\mbbC^8$ which is a composition of projections $P_0P_{a_k}\ldots P_{a_1}P_0$. One could start applying the projections in order and the undesired result at each step is $1-P_{a_i}=P_{-a_i^*/||a_i||^2}$. Assuming this happens, there is an obvious procedure to get to try again the measurement $\{P_{a_i},I-P_{a_i}\}$ until one succeeds: First, we try to project to $P_{a_{i-1}}$. It may succeed and we will have a sequence $P_{a_{i-1}}(I-P_{a_i})P_{a_{i-1}} \propto P_{a_{i-1}}$, therefore getting back where we started. It may not succeed and nevertheless, we will retry $P_{a_i}$, then there are two cases:
\begin{itemize}
    \item Projection is unsuccessful, and we get the sequence
    $$(I-P_{a_i})(I-P_{a_{i-1}})(I-P_{a_i})P_{a_{i-1}}\propto (I-P_{a_{i}})P_{a_{i-1}}.$$
    This means we can start at our first failure.
    \item Projection is successful, then the sequence of projections is 
    $$P_{a_i}(I-P_{a_{i-1}})(I-P_{a_i})P_{a_{i-1}}$$
    which simplifies to 
    $$(\alpha_{i,i-1}^2-1) P_{a_i}P_{a_{i-1}}$$
    which is the desired outcome (notice normalization is allowed as the final unitary gate is considered and $0<\alpha_{i,i+1}<1$).
\end{itemize}
Hence, the only way this process could be unsuccessful is if the angle between $P_{a_i},P_{a_{i-1}}$ is exponentially close to $\frac{\pi}{2}$ or in other words, $\alpha_{i,i-1}$ is exponentially small. This would mean that the composition $(1-P_{a_i})P_{a_{i-1}}$ has to happen exponentially many times before there is a chance of getting a successful outcome.

First, notice that all $a_i$ are chosen from the selection $\{e_i,ce_j,0\}$, which are subsequent projections in (\ref
{Gates}). It needs to be checked for this selection, when exactly the dihedral angle could be exponentially close to $\frac{\pi}{2}$.
\begin{lem}
The dihedral angle between
\begin{itemize}
    \item $P_{e_i}$ and $P_{ce_j}$ is $\frac{\pi}{4}$,
    \item $P_{e_i},P_0$ is $\frac{\pi}{4}$,
    \item $P_{ce_j},P_0$ is $\arccos(\frac{1}{\sqrt{1+c^2}})=\tan^{-1}(|c|)$.
\end{itemize} 
\end{lem}
\begin{proof}
The scalar determining the dihedral angle between the two projections is the scalar given by (\ref{angle_P_1_P_2},\ref{diagonal_blocks})
\begin{align}
\frac{(1+cE_j^\dagger E_i)(1+cE_i^\dagger E_j)}{(1+||E_i||^2)(1+||cE_j||^2)}
\end{align}
which using the following identities 
\begin{itemize}
    \item $||E_j||^2=1, E_i^2=-1$ for $i \neq 0$,
    \item $E_i^\dagger =-E_i$ for $i \neq 0$,
    \item $E_iE_j= -E_jE_i$, for $i,j\neq 0$,
\end{itemize}
is equal to $\frac{1}{2}$. Although $P_0,P_{e_i}$ also form a $\frac{\pi}{4}$ angle (choose $c=0$ above), this is not the case for $P_{ce_j}$ and $P_0$, where choosing $0$ instead of $E_i$ above gives $\frac{1}{1+c^2}$.
\end{proof}
Hence, the above issue \textit{mostly} does not arise simply because all angles between $P_{e_i}$ and $P_{ce_j}$ for $i\neq j$ is equal to $\frac{\pi}{4}$, hence $\alpha^2=\frac{1}{2}$ and the chance of successful outcome is always $\frac{1}{2}$. But for any part of the sequence which is of the form $\ldots P_0P_{ce_j} \ldots $ or $\ldots P_{ce_j}P_0\ldots$, we have to deal with this issue.

A key property of our basis $\{h_i\}$ for $\mathfrak{so}(8)$ is that $h_i$s are not only skew-symmetric matrices but also $h_i^2=-I$. This implies that $\exp(th)=\cos(t)+\sin(t)h=\cos(t) (1+\tan(t)h)$, where $h\in \{h_i\}$ behaves just like the imaginary $i$. Hence, using the same notations in \ref{octonion_universality}, for any $t_i$ in $\prod_l \text{exp}(t_lE_{k_l}E_{j_l})$, one can assume $t_i \in [\frac{-\pi}{2},\frac{\pi}{2}]$. And if $\tan(t_i)=c_i$ is exponentially large, then surely $\tan(t_i/2) \le 1$ is not, as $t_i/2 \le \pi/4$. So the gate $\text{exp}(t_iE_{k_i}E_{j_i})$ in $\prod_l \text{exp}(t_lE_{k_l}E_{j_l})$ can be replaced by two copies of $\text{exp}(\frac{t_i}{2}E_{k_i}E_{j_i})$. Then the result would be to use two copies of the projection sequence $P_0P_{c_ie_{j_i}}P_{e_{k_i}}P_0$ where this time, $c_i$ is $\tan(t_i/2)$ which is no longer exponentially large. 

All in all, the probability of failure in applying our desired sequence of projections is exponentially suppressed in polynomially many steps, which is what is sufficient to be able to claim an efficient implementation of an $SO(8)$ gate, like in the topological forced measurement in \cite{bonderson2008measurement}.

We have established the ability to apply efficiently local gates. The last step is to build the general computational model. ~
\\
\subsection{Universal forced measurement model from octonions}\label{subsection:octonion_model}~
\\
Before describing the model, notice that the last step in any quantum computational machine is a destructive measurement. Equiangular projections are precisely designed not to leak information. Hence we are forced to assume the existence of such a measurement in the model. Here, it will be on the $\{\ket{0},\ket{1}\}$ basis for one qubit. Further, there must be an assumption on the initialization of computational state; here, it will be all qubits set to $\ket{0}$. Then the question becomes how to produce the unitary gates involved in a \textbf{BQP} algorithm. We describe two mechanisms.

The simplest model is to assume qubits $\mathbb{C}^2$ on a circle, where gates only act on adjacent qubits by $SU(4)$. An additional qubit is assumed which only makes sure the computation is in $P_0$ subspace by being in the state $\ket{0}$ (or $\ket{1}$ otherwise).

Notice any algorithm in \textbf{BQP} can be realized by acting on adjacent qubits on the circle with $SU(4)$ gates. For any gate $A=A_1+iA_2 \in SU(4)$, there is the following embedding in $SO(8)$
\begin{align}
\begin{pmatrix}
A_1 & -A_2\\
A_2 & A_1
\end{pmatrix},
\end{align}
which agrees with the encoding $\ket{00} \to e_0, \ket{01}\to e_2, \ket{10} \to e_4, \ket{11}\to e_6$ and $i$ (the imaginary) taking each of these to $e_1,e_3,e_5,e_7$. So any $SU(4)$ gate can be generated by a sequence of projections as an $SO(8)$ gate. Thus, given an algorithm in \textbf{BQP} with $SU(4)$ gates on the neighbor qubits, there is a corresponding sequence of projections giving the $SU(4)$ gates.

In the above model as one moves from an adjacent pair to the next pair, there are two actions of $P_{e_i}$ and they do not necessarily agree on the common qubit, unless we change the encoding accordingly.

There is another more complicated model but with a fixed encoding of the octonions on the qubits, with $2n$ qubits on a circle, each qubit connected to its adjacent ones, and one on the center connected to all qubits. We alternatively number each qubit on the circle by $1,2$ and the qubit at the center $3$. Next, consider the encoding of the vector corresponding to $e_i$ in $\mathbb{C}^8$ as the binary representation of $i$, where the $j$-th location is encoded in the state of a qubit labelled by $j$. For example $e_7$ is encoded as $\ket{111}$ where the first $1$ from the right is encoded in the qubit labelled $1$ and the last in the qubit labelled $3$. Therefore, we have a triangulation of the circle and each triangle is a local representation of $e_i$s (as vector) and of the action of $E_i$s, which is exactly as represented in (\ref{octonion_matrix}). An obvious property about these local representations is that they are all the same (by identifying the qubits with the same label), and the restriction of the action of $E_i$ on two qubits, say $1,3$, is the same for the two actions corresponding to the two triangles containing these two qubits.

An additional qubit labelled by $4$ and connected to all qubits is needed. Similar to the previous model, this qubit will represent whether the computation is done inside the $P_0$ subspace or not, where $P_0=\ket{0}\bra{0}$.

We want to show the universality of the model. As mentioned previously, any algorithm \textbf{BQP} can be done by acting on adjacent qubits on the circle with $A\in SU(4)$ gates. These gates are embedded in $SO(8)$ as in the previous model. Let us call the embedding $A_{new}$.

$A_{new}$ is acting on the same two qubits $A$ acted on, plus the third qubit which is at the center of the circle. Basically, the imaginary part of the computation is encoded in the state of the qubit at the center. This is consistent with how $e_i$s were encoded in the qubits. As all $A_{new} \in SO(8)$ can be generated by a sequence of projections (each acting on $4$ qubits labelled $1,2,3,4$), with the state of qubit $4$ playing the role for the subspace $P_0$, the computational model can produce the unitary gates in any \textbf{BQP} algorithm.

\section{Forced measurement computing model}\label{Forced measurement computational model_section}
In this section, an abstract framework of a forced measurement model is proposed. Like in the quantum unitary circuit model, where one assumes a fixed set of implementable local unitary gates that is universal, we would have a similar collection of universal local measurements for a forced-measurement model. These measurements would be applied to some \textit{initial} state according to an efficient algorithm in order to solve a \textbf{BQP} problem. 

The initial state can be a highly entangled state, like cluster states, or they can be simply tensor product states. This state depends on what resources the computational model in question has at its disposal to create the initial state. We consider this part of the computational model as more of a black-box and do not make assumptions on it other than the obvious requirement that the initial state be created efficiently. But, similar to the standard quantum circuit model, we do assume that the Hilbert space has a tensorial structure. Even in the case of topological forced measurement, where the Hilbert space does not have a tensorial structure, by some encoding one can turn the computations into a tensorial setting, similar to how topological quantum computation by braiding is shown to be universal.

Finally, for the destructive measurement that is to be carried at the end of the computation, since equiangularity prohibits leakage of information, one has to assume an additional local projection on a fixed number of qubits which will serve the purpose of the last measurement. This is also a black-box like the initial state, and the projection depends on the resources of the model. In the octonion case, a spin measurement on a single qubit was assumed but the measurement could have been assumed to be a sum of spin measurement on any number of fixed qubits.

Before we give the abstract formulation of a forced measurement model, let us analyze some examples first.
~
\\
\subsection{Measurement-based quantum computation with cluster states}\label{cluster_states_example}~
\\
The measurement model using cluster states was defined in \cite{raussendorf2001one}, where an entangled state is presumed as the initial state, and no site would be affected twice by a measurement (until the very last destructive measurement). Instead, any outcome of a measurement would tell what the next measurement should be on the next \textit{site}. A site is a local spot of the lattice, hence measurements are clearly applied locally from a fixed set of qubit measurements which can be thought of as the \textit{set of local forced measurements}. 

While this may look unlike a forced measurement model, but it is still within that framework. Indeed, it is a more ideal version of a forced measurement where there is no probability (zero) of failing in the measurements, as there is no case of failure \textit{by design}. Notice, this probability is non-zero in the topological or octonion model, but it is exponentially suppressed. 

On the other hand, the cluster state measurement model never tries to go back and \textit{try again} the previous measurement. Further, the measurements applied on the sites are from a set of fixed qubit measurements which are obviously strongly equiangular as lines in $\mathbb{C}^2$ are always so.~
\\
\subsection{Measurement-only topological quantum computation}\label{TQC_example}~
\\
Forced topological quantum measurement was introduced in \cite{bonderson2008measurement}, where braiding, which is the unitary gate in topological quantum computation, is expressed as the composition of three projections. For the details and background on graphical calculus, we refer to \cite{bonderson2009measurement}. 

This model is actually a more complicated version of a forced measurement model like the octonion model. We explain briefly how the computation works. 

The model stores the processed information $\psi$ in the fusion of anyons and uses two ancillas $a$ and its antiparticle $a^*$ to perform forced measurement. Therefore there is a sequence of anyons $a$ with two anyons $a^*,a$ in-between each two anyons.
\begin{figure}[h]
\begin{center}
  \includegraphics[scale=0.4]{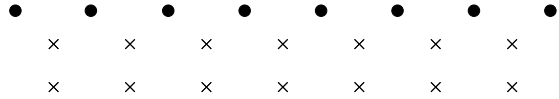}
  \caption{Ancillas denoted by X's are between adjacent computational anyons (denoted by dots). Figure from \cite{bonderson2008measurement}.}
  \label{fig:quasiarray}
\end{center}
\end{figure}

Locally, at any time, the computation is done on the following fusion tree:
\begin{align}\label{psi_information}
\pspicture(-1,1.0)(1.5,2.5)
  \small
  \psset{linewidth=0.9pt,linecolor=black,arrowscale=1.5,arrowinset=0.15}
  \psline(0.0,0.5)(1.2,2)
  \psline(0.0,0.5)(-1.2,2)
  \psline(-0.4,1)(0.4,2)
  \psline(-0.8,1.5)(-0.4,2)
    \psline{->}(-0.4,1)(0.3,1.875)
    \psline{->}(0.8,1.5)(1.1,1.875)
    \psline{->}(-0.8,1.5)(-0.5,1.875)
    \psline{->}(-0.8,1.5)(-1.1,1.875)
    \psline{->}(0,0.5)(-0.3,0.875)
    \psline{->}(0,0.5)(-0.7,1.375)
  \psline(0,-0.5)(0.0,-0.15)
  \psline(0,0.3)(0,0.5)
   \psframe[linewidth=0.9pt,linecolor=black,border=0](-0.25,0.3)(0.3,-0.15)
  \rput[bl]{0}(-0.1,-0.1){$\psi$}
  \rput[bl]{0}(-0.9,1){$c$}
  \rput[bl]{0}(-0.5,0.5){$e$}
  \rput[bl]{0}(-1.35,2.1){$a$}
  \rput[bl]{0}(-0.5,2.1){$a^*$}
  \rput[bl]{0}(0.3,2.1){$a$}
  \rput[bl]{0}(1.15,2.1){$a$}
  \scriptsize
  \rput[bl]{0}(-1.05,1.35){$\alpha$}
  \rput[bl]{0}(-0.65,0.75){$\beta$}
 \endpspicture,
\end{align}
\vskip .5in
where the middle two anyons are the ancillas. The labels $\beta,\alpha$ belong to the Hom spaces $\text{Hom}(e,a\otimes c),\text{Hom}(c,a\otimes a^*)$, respectively. The three projections used to perform the braiding of the first and last anyon $a$ above, each consists of composition of projections which fuse $a,a^*$ to the vacuum. The ultimate result is the following braiding (\cite{bonderson2008measurement}) which is unitary on $\psi$:
\begin{align}
\pspicture(-1,1.0)(1.5,2.5)
  \psset{linewidth=0.9pt,linecolor=black,arrowscale=1.5,arrowinset=0.15}
  \psline(-0.4,0)(0,0.5)
  \psline(0.4,0)(0,0.5)
  \psline(1.2,0)(-1.2,2)
  \psline[border=2.5pt](-1.2,0)(1.2,2)
  \psline(0,1.5)(-0.4,2)
  \psline(0,1.5)(0.4,2)
  \psline{-<}(0,0.5)(0.3,0.125)
  \psline{-<}(0,0.5)(-0.3,0.125)
  \psline{->}(0,1.5)(-0.3,1.875)
  \psline{->}(0,1.5)(0.3,1.875)
  \psline{-<}(-1.2,2)(-0.6,1.5)
  \psline{-<}(1.2,2)(0.6,1.5)
   \rput[bl]{0}(-1.3,2.1){$a$}
   \rput[bl]{0}(-0.5,2.1){$a^*$}
   \rput[bl]{0}(0.35,2.1){$a$}
   \rput[bl]{0}(1.15,2.1){$a$}
   \rput[bl]{0}(-1.3,-0.3){$a$}
   \rput[bl]{0}(-0.5,-0.3){$a^*$}
   \rput[bl]{0}(0.35,-0.3){$a$}
   \rput[bl]{0}(1.15,-0.3){$a$}
 \endpspicture
\end{align}
\vskip 0.5in
Diagrammatically, the fusion to vacuum is presented by the Temperley-Lieb algebra operators $e_i$:
\begin{align}\label{e_i_diagram}
e_i=\frac{1}{d_a}
    \pspicture(-1,0.0)(1,1.5)
  \small
  \psset{linewidth=0.9pt,linecolor=black,arrowscale=1.5,arrowinset=0.15}
  \psline(0.0,0.25)(0.6,1)
  \psline(0.0,0.25)(-0.6,1)
    \psline{->}(0.4,0.75)(0.55,0.9375)
    \psline{->}(-0.4,0.75)(-0.55,0.9375)
  \psset{linewidth=0.9pt,linecolor=black,arrowscale=1.5,arrowinset=0.15}
  \psline(0.0,-0.25)(0.6,-1)
  \psline(0.0,-0.25)(-0.6,-1)
    \psline{-<}(0.4,-0.75)(0.55,-0.9375)
    \psline{-<}(-0.4,-0.75)(-0.55,-0.9375)
  \rput[bl]{0}(-0.7,1.075){$a^*$}
  \rput[bl]{0}(0.6,1.075){$a$}
  \rput[bl]{0}(-0.7,-1.25){$a^*$}
  \rput[bl]{0}(0.6,-1.25){$a$}
 \endpspicture,
\end{align}
\vskip 0.5in
where $d_a$ is the quantum dimension of anyon $a$. These projections satisfy the famous identities 
\begin{align}
e_{i \pm 1} e_i e_{i \pm 1}=\frac{1}{d^2}e_{i\pm 1}, \ \ \  [e_i,e_j]=0, \ \forall |i-j|>1,    
\end{align}
where $d=d_a$ also corresponds to the loop value in the diagrammatic formulation of Temperley-Lieb algebra (\ref{e_i_diagram}), as two $e_i$s are stacked. These are exactly the identities we need for the local forced measurements $e_i$, which are equiangular. Notice they are not strongly equiangular as the similar identity does not hold for $e_i$ replaced by $1-e_i$. Also, the projections acting on disjoint sites commute.

If the projection $e_i$ does not succeed then according to the fusion rules of anyons for $a \otimes a^*= \sum_{i=1}^k N_{aa^*}^{x_i}x_i$ (with $x_0=1$, the vacuum), the outcome is another projection to an anyon $x_i\neq 1$. Therefore the measurements are no longer of the type $\{P,I-P\}$ but with multiple possible outcomes $\{P_1,P_2,\ldots,P_k\}$ where $P_i$ is fusion into $x_i$.
\begin{align}\label{P_i_diagram}
P_i=\frac{\sqrt{d_{x_i}}}{d_a}
    \pspicture(-1,0.0)(1,1.5)
  \small
  \psset{linewidth=0.9pt,linecolor=black,arrowscale=1.5,arrowinset=0.15}
  \psline(0.0,0.25)(0.6,1)
  \psline(0.0,0.25)(-0.6,1)
    \psline{->}(0.4,0.75)(0.55,0.9375)
    \psline{->}(-0.4,0.75)(-0.55,0.9375)
  \psset{linewidth=0.9pt,linecolor=black,arrowscale=1.5,arrowinset=0.15}
  \psline(0.0,-0.25)(0.6,-1)
  \psline(0.0,-0.25)(-0.6,-1)
    \psline{-<}(0.4,-0.75)(0.55,-0.9375)
    \psline{-<}(-0.4,-0.75)(-0.55,-0.9375)
  \psline(0,-0.25)(0,0.25)
  \psline{->}(0,-0.25)(0,0.075)
  \rput[bl]{0}(-0.7,1.075){$a^*$}
  \rput[bl]{0}(0.6,1.075){$a$}
  \rput[bl]{0}(-0.7,-1.25){$a^*$}
  \rput[bl]{0}(0.6,-1.25){$a$}
  \rput[bl]{0}(0.1,-0.05){$x_i$}
 \endpspicture,
\end{align}
\vskip 0.5in

Also, it can be observed that for two consecutive sites on the fusion tree (\ref{psi_information}), the corresponding projections are no longer all pairwise equiangular. Still, the protocol works as equiangularity holds when restricted to a \textit{protected} subspace, as explained below.

This motivates us to relax a requirement that one may have assumed about a forced measurement model: only using equiangular consecutive projections for the computation. Indeed, we need to allow the use of consecutive non-equiangular projections as long as it is used to get us back from where we started in the case of an undesired outcome. But one needs to make sure that even in that case, no information is being leaked from a subspace, called $\mathcal{K}$, where the processed information is kept. Therefore the additional requirement is that the consecutive application of possibly non-equiangular projections $\prod P$ is still a unitary on $P_{\mathcal{K}}$. In other words $P_{\mathcal{K}} (\prod P) P_{\mathcal{K}} \propto (U \oplus I_{\mathcal{K}^\perp}) P_{\mathcal{K}}$ for some unitary $U$, ideally identity, so that it would be easiest to interpret the final outcome of the computation. This is similar to Quantum Turing Machine, where one needs to allow the model to use all the advantages of having \textit{ancillas}.

In the case of \cite{bonderson2008measurement}, the projection $P_{\mathcal{K}}$ is the sum of the projections in  (\ref{fusion_to_c}), as we only care about the fusion of the relevant anyons where the processed information $\psi$ in (\ref{psi_information}) is stored, not the information given by fusion of ancillas.
\begin{align}\label{fusion_to_c}
P_{c}=\frac{\sqrt{d_{c}}}{d_a^2}
    \pspicture(-1.5,.0)(1.75,2.5)
  \small
  \psset{linewidth=0.9pt,linecolor=black,arrowscale=1.5,arrowinset=0.15}
  \psline(0.0,0.5)(1.2,2)
  \psline(0.0,0.5)(-1.2,2)
  \psline(0.0,1.5)(0.4,2)
  \psline(0.0,1.5)(-0.4,2)
    \psline{->}(0.0,1.5)(0.32,1.9)
    \psline{->}(0.8,1.5)(1.1,1.875)
    \psline{->}(0.0,1.5)(-0.32,1.9)
    \psline{->}(-0.8,1.5)(-1.1,1.875)
  \psset{linewidth=0.9pt,linecolor=black,arrowscale=1.5,arrowinset=0.15}
  \psline(0.0,-0.5)(1.2,-2)
  \psline(0.0,-0.5)(-1.2,-2)
    \psline{-<}(0.8,-1.5)(1.1,-1.875)
    \psline{-<}(-0.8,-1.5)(-1.1,-1.875)
    \psline(0.0,-1.5)(0.4,-2)
    \psline(0.0,-1.5)(-0.4,-2)
    \psline{-<}(0.0,-1.5)(-0.32,-1.9)
    \psline{-<}(0.0,-1.5)(0.32,-1.9)
  \psline(0,-0.5)(0,0.5)
  \psline{->}(0,-0.5)(0,0.1)
  \rput[bl]{0}(0.15,-0.2){$c$}
  \rput[bl]{0}(-1.3,2.1){$a$ \,\,\,\,\, $a^*$ \,\,\,\,\, $a$ \,\,\,\,\, $a$ }
  \rput[bl]{0}(-1.3,-2.3){$a$ \,\,\,\,\, $a^*$ \,\,\,\,\, $a$ \,\,\,\,\, $a$ }
 \endpspicture, \text{where } a \otimes a = \sum_c N_{aa}^c c
\end{align}
\vskip 0.8in
As long as projections do not make the first and last anyons to directly or indirectly fuse, e.g. a sequence of projections only on the second, third and last anyon, information is not leaked from $\psi$. Non-equiangular projections are precisely composed in this manner.~
\\
\subsection{Local forced measurement model}~
\\
The above examples guide us towards what the definition of a \textit{local} forced measurement model should be. 

First we need to fix a set of measurements. We start with a set $\mathcal{M}$ of local measurements which may or may not be equiangular. Recall a projective operator-valued measure (POVM) $\mathcal{M}$ over an $n$-dimensional Hilbert space is a set of $m$ positive semi-definite operators like $M_i, 1\le i \le m$, such that 
\begin{align}\label{sum_to_identity}
\sum_{i=1}^m M_i= I_{n\times n}.
\end{align}

\begin{dfn}\label{forced_measurement_graph_dfn}
Define a local forced measurement model as a triple $(\mathcal{M},\mathcal{H},\mathcal{K})$ of POVMs $\mathcal{M}$ acting on a fixed Hilbert space $\mathcal{H}=(\mathbb{C}^2)^{\otimes c}$ with a tensorial structure, and a distinguished subspace $\mathcal{K} \subset \mathcal{H}$ with corresponding projection $P_{\mathcal{K}}$.
\end{dfn}

A measurement does not always have the desired outcome. Hence, it is necessary to establish what is meant by a sequence of \textit{adaptive} measurements which would give the desired outcome.
\begin{dfn}
A \textit{sequence of adaptive measurements} is a sequence in which the choice of each measurement is dependent on the previous outcomes and this choice is determined in at most classical polynomial time with respect to the length of the sequence.
\end{dfn}

Notice each sequence has an associated probability as we are performing measurements. As in a computation there is a restriction on the length of the sequence, one has to consider the total probability of \textit{desired} sequences. This would determine if one can do a computation efficiently, similar to the octonion model where the failure probability is suppressed.
\begin{dfn}\label{universal_G}
The local model $(\mathcal{M},\mathcal{H},\mathcal{K})$ is called \textit{universal} if for every unitary gate $U$ on $\mathcal{K}$ and given $\epsilon$, with probability more than $\frac{2}{3}$, one can approximate the following gate up to some scalar  
\begin{align}
(U\oplus I_{\mathcal{K}^\perp})P_{\mathcal{K}},
\end{align}
with error $\epsilon$, by a sequence of adaptive measurements in $\mathcal{M}$, whose length is at most $poly(\frac{1}{\epsilon})$  for some fixed polynomial $poly$.
\end{dfn}

The above definition is very similar to the universal quantum computation basis definition where local unitary gates are used \cite{kitaev1997quantum}.
\begin{rmk}
For the octonions, \textit{exact} generation of the unitary gates was achieved. Hence, one can first generate exactly a well-known set of universal unitary gates like the $\{$CNOT, Hadamard, $\frac{\pi}{4}$-phase-shift, $\frac{\pi}{2}$-phase-shift$\}$ gates, and as these satisfy the property above (\cite{kitaev1997quantum}), the octonions give a universal local model.
\end{rmk}
\begin{rmk}
Although the adaptive sequence is not assumed to come from equiangular projections, but it is usually the case that $\mathcal{M}$ has a collection of equiangular projections as seen in the examples \ref{cluster_states_example} and \ref{TQC_example}. These projections are the ones that give the unitary gates. Hence, equiangular projections can be seen as the most likely tool to be used during a forced measurement computation, and strongly equiangular projections as the most ideal tool.
\end{rmk}

As in any definition of a computational model, there is some flexibility in the definition. For example, we can embed $\mathcal{H}$ isometrically inside a bigger space and still have the same measurements with (non-)equiangularity preserved. So the dimension of the qubits factor for $\mathcal{H}$ can be any constant $>1$. 

As discussed before, the initial state and measurement part of the computational model are black-boxes and should therefore be analyzed in the specific context. That is why we add the following additional structures to the local forced measurement model:
\begin{dfn}\label{LFMM_dfn}
A local forced measurement machine (LFMM), conveniently called by its set of local measurements $\mathcal{M}$, consists of 
\begin{itemize}
    \item a local forced measurement model $(\mathcal{M},\mathcal{H},\mathcal{K})$,
    \item a measurement process $P_Z$ acting on $\mathcal{H}$ serving the role of final measurement,
    \item a sequence of states $\{\ket{\psi}_n\}_{n\in \mathbb{N}}$ serving the role of the initial state for input of size $n$, with $\ket{\psi}_n \in \mathcal{H}^{\otimes L(n)}$ for some polynomially bounded function $L:\mathbb{N} \to \mathbb{N}$ and a polynomially bounded function $q:\mathbb{N} \to \mathbb{N}$ which describes the preparation time of $\ket{\psi}_n$.
    \item For each $n$, There is a graph $H(n)$ with $cL(n)$ vertices representing qubits in $\mathcal{H}^{\otimes L(n)}=(\mathbb{C}^2)^{\otimes cL(n)}$, and a set of $\{H_i\}$ of subgraphs with $c$ vertices on which projections of $\mathcal{M}$ act. The representation of these projections on $\{H_i(n)\}$ should be \textit{consistent}, meaning their restriction on $H_i(n) \cap H_j(n)$ should be the same no matter from which subgraph restriction is made. This will serve as the architecture of the circuit for inputs of size $n$.
\end{itemize}
\end{dfn}

Next, the forced measurement computation for a local model is defined. We follow the notations used in \cite{arora2009computational}.  $\{0,1\}^*$ means all finite strings in $\{0,1\}$, while $f:\{0,1\}^* \to \{0,1\}$ and $T:\mathbb{N} \to \mathbb{N}$ denote a problem and a time constructible function used to measure the running time of the algorithm. This definition is an analog of a quantum circuit with local unitary gates.

\begin{dfn}\label{LFMM_computation}
Let $f$ be a problem, $T$ a time constructible function, and $\mathcal{M}$ be an LFMM as in \ref{LFMM_dfn}. The problem $f$ is said to be computable in $(T(n)+q(n))$-time for $\mathcal{M}$, if for every input $x$ of size $n$, there is a sequence of adaptive measurements of length at most $T(n)$ such that
\begin{enumerate}
    \item each projection is one of the projections of $\mathcal{M}$ acting on a subgraph $H_i(n) \in \{H_i(n)\}$ of $H(n)$,
    \item no leakage occurs from the subspaces $\mathcal{K}$ of each subgraph,
    \item the sequence followed by the measurement $P_Z$ on one of the $H_i(n)$s gives $f(x)$ with probability at least $\frac{2}{3}$.
\end{enumerate}
\end{dfn}
\begin{rmk}
In the topological model, we mentioned that using non-equiangular projections $\prod P$ should not leak information from the protected space $\mathcal{K}$ or in other words, $\prod P$ is a unitary on $P_\mathcal{K}$. This condition is the item (2) above.
\end{rmk}
\begin{rmk}
There is also a definition for a general forced measurement model, without the local representation restrictions. More precisely, the notion of the subgraphs and their accompanying restriction on the local representations can be entirely removed. This gives a more succinct description of LFMM, an example being the first model in \ref{subsection:octonion_model}. But from a practical point of view, it is the models with the definition above that we have to usually deal with, where measurements can only be applied in a very specific way and there is only one fixed local representation given, as in the second model in \ref{subsection:octonion_model}.
\end{rmk}
We could go even further and omit the fixed set of local forced measurements. In that case, we have to ensure that in the above definition, each projection's description is given by a polynomial-time classical Turing Machine, just like in the description of a Quantum Turing Machine \cite[Def. 10.9]{arora2009computational}. This gives the most general form of forced measurement computation. The definition is an analog of QTM definition \cite[Def. 10.9]{arora2009computational}.
\begin{dfn}
Let $f$ be a problem, $T$ a time constructible function. The problem $f$ is said to be computable in \textit{forced measurement} $T(n)$-time, if for every input $x$ of size $n$, there is a sequence of adaptive measurements of length at most $T(n)$ such that:
 \begin{enumerate}
    \item the initial state is $\ket{x0^{n+T(n)}}$ ($x$ padded with zeros),
    \item each measurement description is given by a fixed classical TM (only dependent on input size $n$) in polynomial-time given the classical input $(1^n,1^{T(n)})$,
    \item the sequence followed by a measurement on the first qubit gives $f(x)$ with probability at least $\frac{2}{3}$.
\end{enumerate}
\end{dfn}

\section{Closely related topics}
In this section, we briefly mention a few closely related topics and leave the detail to the future.~
\\
\subsection{Adiabatic quantum computation}~
\\
Adiabatic Quantum Computation (AQC) is well-known to be universal \cite{aharonov2008adiabatic}. The idea in AQC is that the slow enough evolution of a system allows the state to stay in the desired eigenspace.

In a forced measurement model, one may be able to use the idea of AQC to suppress even more the probability of failure. As an example, as was described in the efficient implementation of local gates in \ref{subsection:efficient_implementation}, instead of replacing the gate $\text{exp}(t_iE_{k_i}E_{j_i})$ in $\prod_l \text{exp}(t_lE_{k_l}E_{j_l})$ by two copies of $\text{exp}(\frac{t_i}{2}E_{k_i}E_{j_i})$, one can replace it with $poly(n)$ copies of $\text{exp}(\frac{t_i}{poly(n)}E_{k_i}E_{j_i})$, thereby lowering the probability of failure from some constant to $\frac{1}{n^r}$ for any fixed $r>0$.

Conversely, AQC can be thought of as a polynomial sequence of very close projections $P_i$ which are nearly equiangular. These projections are the ground space of the Hamiltonians at time $t_i$. Their closeness suppresses exponentially the probability of failure.~
\\
\subsection{SIC-POVM: a possible generalization}~
\\
Recall the definition of POVM \ref{sum_to_identity} for $\mathcal{M} = \{M_i\}_{i=1}^m$ acting on $n$-dimensional Hilbert space. If $m\ge n^2$, and $M_i$ span the linear map space $\mathcal{L}(\mathcal{H})$, then $\{M_i\}$ is called a collection of informationally complete POVM (IC-POVM).
 
A \textit{minimal} IC-POVM happens when $m=n^2$ and it is symmetric (SIC-POVM) (\cite{zauner1999quantum}) if $M_i$ are a $\lambda$-scaled rank one projections such that Hilbert-Schmidt inner product $\text{Tr}(M_iM_j)$ is a constant $c$ for all pairs. $\lambda$ turns out to be $n$ and $c$ turns out to be $\frac{1}{n+1}$. In other words, after normalization:
\begin{align}
\text{Tr}(M_iM_j)=\frac{n\delta_{ij}+1}{n+1}.
\end{align}

Rank one projections are always equiangular. The above definition can be modified using strongly equiangular projections as dihedral angle and the Hilbert Schmidt inner product are related.
\begin{dfn}
A strongly equiangular IC-POVM is a minimal IC-POVM which are strongly equiangular projections up to some scale where all pair-wise dihedral angles are equal, in other words $M_iM_jM_i=\alpha^2M_j$ for a fixed $\alpha$.
\end{dfn}
We know that for equiangular projections $\text{Tr}(PQ)=\text{Tr}(PQQ)=\text{Tr}(QPQ)=\alpha^2 \text{rank}(Q)$. Computing $\text{Tr}((\sum_{i=1}^m M_i)^2)=\text{Tr}(I^2)=n$ by using (\ref{sum_to_identity}), and noting that $\alpha^2$ is the only (non-zero) eigenvalue of $M_iM_j$, we have
\begin{align}
\alpha^2=\frac{(n^2-2)}{2(n^2-1)}.
\end{align}
Notice as $n\to \infty$, $\alpha \to \frac{1}{\sqrt{2}}$, which means the planes would have to meet at $\frac{\pi}{4}$ degrees.
\begin{ques}
For which $d$ there exists strongly equiangular IC-POVM?
\end{ques}~
\\
\subsection{Error correcting codes}~
\\
Consider a set of linear errors $\mathcal{E}$. In most cases, one can assume the errors are generated by Pauli operators. It is well-known that a subspace $\mathcal{C} \subseteq \mathcal{H}$ is a quantum error correcting code if
$$ C(E)=\bra{\psi}E\ket{\psi},\  \forall E \in \mathcal{E}, \ket{\psi} \in \mathcal{C},$$
where $C(E)$ does not depend on $\ket{\psi}$. A stronger version of the above identity is when $C(E)P_{\mathcal{C}}=P_{\mathcal{C}}EP_{\mathcal{C}}$.

For example in the perturbed Toric code \cite{kitaev2003fault}, or generally codes based on topological phases of matter (see e.g. \cite{bonderson2013quasi}) on a lattice, one uses the fact that the local errors can be corrected as $P_1 E P_2$ is exponentially small for projections $P_1 \neq P_2$ on two different orthogonal states in the code. In other words, $EP_{\mathcal{C}}$ and $P_{\mathcal{C}}$ are two planes with all dihedral angles almost equal and in fact very close to $\frac{\pi}{2}$ as the lattice size increases.

Therefore a \textit{weaker} notion of equiangularity, where only one of the equations $PQP=\alpha^2 P,QPQ=\alpha^2 Q$ holds, is a desirable condition for a quantum error correcting code. Except that in this context, $Q$ is not a projection but perhaps a Pauli operator.~
\\
\subsection{Equiangular projections by optimal packings}~
\\
A source for a discrete collection of strongly equiangular planes is the solution to the problem of optimal packing in Grassmannian manifolds. It was shown in \cite{shor1998family} that an optimal family of $m^2+m-2$, $\frac{m}{2}$-dimensional planes exists in $\mathbb{R}^m$, when $m$ is a power of two.

A family of planes which is an optimal packing family is expected to achieve optimal minimum \textit{distance}, where distance between two planes is defined by 
$$\sqrt{\sum_{i=1}^{k} \sin(\theta_i)^2},$$
where $\theta_i$ are the dihedral angles between the two planes.

One would wonder if it is possible to obtain $O(2)$ gates using the optimal packing family for $m=4$; notice the planes live in $\mathbb{R}^4$, so it is not possible to aim for $U(2)$. The planes in an optimal packing are not mutually strongly equiangular, and every plane is only strongly equiangular to $\frac{m^2}{2}$ of the other $m^2+m-3$ planes. Further, the dihedral angle is always $\frac{\pi}{4}$. 

Each plane can be represented by a $\frac{m}{2}\times m$ matrix  with rows being an orthogonal basis for the plane. Let us call these matrices $\{M_k\}$. For $m=4$, they are all given in \cite[Eq. (3)]{shor1998family}. The matrices are seen to have $\pm1,0$ entries with rows having the same norm: either $1$ or $\sqrt{2}$. The projections onto these planes is given by $\{P_k=\frac{1}{\alpha_k^2}M_k^\dagger M_k\}$ where $P_0=M_0^\dagger M_0$ is
\begin{align}
    \begin{pmatrix}
    I_{2\times 2} & 0 \\
    0  & 0
    \end{pmatrix}.
\end{align}
Next, a sequence $P_0P_kP_{k-1}\ldots P_1P_0$ can be computed as
\begin{align}
\frac{1}{\prod_{j=1}^k \alpha_j^2}M_0^\dagger(M_0M_k^\dagger)\ldots (M_1M_0^\dagger)M_0,
\end{align}
where each $\frac{1}{\alpha_j\alpha_{j-1}}M_jM_{j-1}^\dagger$ is precisely the $2 \times 2$ unitary up to scale matrix in Remark \ref{angle_P_1_P_2}. Indeed, there exist always $m \times m$ unitaries $U_j$ such that $\frac{1}{\alpha_j}M_j=M_0U_j^\dagger$. Therefore, $P_j=U_jM_0^\dagger M_0 U_j^\dagger$ being consecutively strongly equiangular implies the top left $2\times 2$ block of $U_j^\dagger U_{j-1}$ is a unitary up to a scale which is in fact $M_0U_j^\dagger U_{j-1}M_0^\dagger=\frac{1}{\alpha_j\alpha_{j-1}}M_jM_{j-1}^\dagger$. 

Now, take any two equiangular planes $M_k,M_j$ in \cite[Eq. (3)]{shor1998family}. The product $M_kM_j^\dagger$ is the dot product of each of the row vectors in $M_k$ and $M_j$. If two rows are non zero at the same locations, then their dot product must be zero (otherwise it means the planes $M_k,M_j$ share a vector, so not equiangular). For example, they must be $(0++0),(0+-0)$. And if they share a single entry then their dot product is $\pm1$. This means the product $M_kM_j^\dagger$ is a matrix with  $\pm 1,0$ entries. Therefore the gates are just a multiple of an orthogonal matrix made of $\pm1,0$, there are all either reflections, (or composed with a) rotation of 45 or 90 degrees, and can only generate a finite subgroup of $O(2)$. 

The general case remains unsolved:
\begin{ques}
Can we get a universal quantum computer for $m>4$?
\end{ques}
Many of the facts mentioned for $m=4$ holds for general $m$ by a simple induction using the recursive definition of optimal packings described in \cite{shor1998family}. It can be shown all rows of matrices $M_k$ have only $\pm1,0$ entries and the number of $\pm1$ entries is the same for all rows for any matrix $M_k$. Hence, the projection $P_k$ can be defined as $\frac{1}{\alpha_k^2}M_k^\dagger M_k$ as the norm of all rows is the same $\alpha_k \in \mathbb{N}$. Then, it can be also demonstrated that any unitary matrix we get from $M_kM_j^\dagger$ is a multiple of a matrix made of $\pm1 ,0$. 

Although for $m=4$, these do not generate a dense subgroup, in general, orthogonal matrices made of $\pm1,0$ can generate $U(\frac{m}{4})$. Indeed, $U(\frac{m}{4})$ embeds into $O(\frac{m}{2})$ as we did for the embedding $U(4) \hookrightarrow O(8)$ in the case of octonions. Further, the group generated by Toffoli gate $T$, Hadamard $H$, and $\frac{\pi}{2}-$phase-shift $E=\binom{1 \ 0}{0 \ i}$ is dense in the unitary group, and all are a multiple of a matrix made of $\pm 1,0$ after the embedding $T \in O(16),H \in O(4),E \in O(4)$.

By taking $m=32$ (and $m=8$), it can be shown that one can obtain $T$ (and $H,E$) using some matrix product $M_kM_j^\dagger$. So what remains to prove is that one can produce any sequence $U_r \ldots U_1$ of these gates.

The unsolved issue is that choosing any matrix $M_k$ restricts the next choices for the projections; recall the same issue in \ref{octonion_universality}. We have not been able to get $T,H,E$ in a way that they be composable, e.g. to get $E$ from a sequence $P_1 \ldots P_k$ and $H$ from $Q_1\ldots Q_{k'}$ and have $Q_{k'},P_1$ equiangular (and $P_k,Q_1$ equiangular).

\section{Chemical Protection}~
\\
A motivation for this paper was the idea that the authors learned from M.~ Fisher, albeit in a very different context \cite{fisher2015quantum}. The idea is that the binding of two molecules can, in some circumstances, implements a projection within their shared nuclear spin Hilbert space. To understand this principle, in a simple context, first consider the isomers of molecular hydrogen $H_2$. They are parahydrogen and orthohydrogen according to the spin state: singlet or triplet, respectively, or the two proton spins:
\begin{align}
\mathbb{C}^2\otimes \mathbb{C}^2 \cong \text{Singlet} \oplus \text{Triplet},
\end{align}
\begin{align}
\text{Singlet}=\ket{\uparrow \downarrow} - \ket{\downarrow \uparrow}, \ \text{Triplet} = \ket{\uparrow \uparrow} \oplus \ket{\downarrow \downarrow} \oplus (\ket{\uparrow \downarrow} + \ket{\downarrow \uparrow}).
\end{align}
The wave function $\psi$ of $H_2$ has both spacial and spin tensor factors which together must obey Fermi statistics under exchange of the protons. This implies that the angular momentum quantum number $\ell$ must be even for parahydrogen and odd for orthohydrogen: orthohydrogen cannot stop tumbling. In fact, it is experimentally observed in liquid hydrogen that, decay processes gradually this kinetic energy to heat as ortho decays to parahydrogen.

Now, fancifully, assume that there was some chemical reason to bond two $H_2$'s side by side, and simultaneously we were able to restrict the angular momentum of each $H_2$ to be the spin of $\ell=0$ and $\ell=1$:
\begin{align}
 \pspicture(-0.5,-0.1)(0.5,0.1)
  \small
  \psset{linewidth=0.9pt,linecolor=black,arrowscale=1.5,arrowinset=0.15}
  \psline(0.0,0.225)(1,0.225)
  \psline(0.0,-0.225)(1,-0.225)
  \pscircle(-0.125,0.225){0.125}
  \pscircle(-0.125,-0.225){0.125}
  \pscircle(1.125,0.225){0.125}
  \pscircle(1.125,-0.225){0.125}
\endpspicture
\ \ \ \ \ \ \ 
\longrightarrow
\pspicture(-0.5,-0.1)(0.5,0.1)
  \small
  \psset{linewidth=0.9pt,linecolor=black,arrowscale=1.5,arrowinset=0.15}
  \psline(0.0,0.125)(1,0.125)
  \psline(0.0,-0.125)(1,-0.125)
  \pscircle(-0.125,0.125){0.125}
  \pscircle(-0.125,-0.125){0.125}
  \pscircle(1.125,0.125){0.125}
  \pscircle(1.125,-0.125){0.125}
\endpspicture
\ \ \ \ \ \ : \ H_2 +H_2 \longrightarrow H_4,
\end{align}
Then it seems reasonable to assume that such a binding would project to the sector where the two angular momentums agree and hence on spin-space it would implement a projection 
\begin{align}
    P: (\mathbb{C}^2)^{\otimes 4} \to \text{Singlet} \otimes \text{Singlet} \oplus \text{Triplet}\otimes \text{Triplet},
\end{align}
which has rank $10$. Similarly, failure to bind would implement the complementary projection $I-P$ of rank $6$. Now imagine a highly controlled gas of (our modified) $H_2$'s, where any pair of molecules can at our instruction be brought together and allowed to bind or not bind, effecting the projections:
\begin{align}
    \text{bind}:\ P, \ \text{not bind}:\ I-P.
\end{align}
If binding occurs we would quickly alter the chemical environment to pull them apart. In this way we can imagine a computer which operates on spin-space by a sequence of observed projections $\{P,I-P\}$ applied at our choice to any sequence of projection (of course one could contemplate parallelizing the sequence to the extent that the pairs be brought together are non-overlapping). 

This computer would be a poor one. We see no equiangularity (even when restricted to a computational subspace $\mathcal{K}$ as in section \ref{Forced measurement computational model_section}). So there would be leakage and there is no hint, as we see, that universal quantum computing would be possible using this hypothetical gas of dimeric molecules, even given our fanciful assumption for their manipulation.

However, this example can be enhanced in many ways, and we hope to investigate whether realistic enhancements might yield \textbf{chemically protected} quantum computers. Our phrase ``chemically protected'' is a deliberate play on ``topologically protected''.

The idea is that a small molecule may have an interesting symmetry group $G$ ($G \cong \mathbb{Z}_2$ in the $H_2$ example), which acts on some subset of its nuclear spin, spanning $\mathcal{H}_i$ for the $i$-th molecule.

Now, $\mathcal{H}_i$ decomposes to $\oplus_{k} \mathcal{H}_{i,k}$ as a sum of irreducible $G$-representations. Binding should correspond to a projection $P$ onto $\sum \mathcal{H}_{i,k}\otimes \mathcal{H}_{j,k'}$ where the sum is taken over pairs $(k,k')$ of irreps compatible with Fermi statistics, as in our example. Thus, binding/not binding implements a projection $P_{i,j}$ or $I-P_{i,j}$ on $\mathcal{H}_i \otimes \mathcal{H}_j$. Note that if multiple inequivalent binding geometries are possible, we should track these with an additional index $P_{i,j}^\alpha$. Another variation of the projections could come from the entanglement between the spins of two molecules which could affect their binding probability.

We ask the question whether a universal quantum computer can be fashioned from these quantum-mechanical projections. The projections $P_{i,j}$ are protected by the rigidity of the small molecule. Deviations from symmetry due to phonons or isotopic variation constitute a source of error. Although there are no exponential scalings as in the theory of topological protection, small molecule rigidity is quite robust and could be expected to be a useful resource. 

The project is \textit{first} to find within the representation theory strong equiangularity (perhaps merely restricted to a computational subspace) within the families $\{P_{i,j}^\alpha\}$, and then \textit{second} translate the representation theoretic solution into chemistry. Many constraints, here, have already been explored \cite{fisher2015quantum}. For example, there are good reasons to use spin$=\frac{1}{2}$ nuclei with ${}^{31}P$ being a prime candidate. The possibility of using Posner molecules has also been explored in \cite{2017arXiv171104801Y}.

\section*{Acknowledgements}

We would like to thank Michael Hopkins for connecting us with the literature on the unstable homotopy groups of the orthogonal group, leading to the results in Remark \ref{S^3_hom} and \ref{SO(n-k)_hom}.

\bibliographystyle{apa}
\bibliography{main}
\address{\textsuperscript{1\label{1}}Station Q, Microsoft Research, Santa Barbara, California 93106, USA and Department of Mathematics, 
University of California, Santa Barbara, California 93106 USA}
\address{\textsuperscript{2\label{2}}Microsoft Station Q and Dept of Mathematics, University of California,
Santa Barbara, CA 93106-6105, U.S.A.} 
\address{\textsuperscript{3\label{3}}Dept of Mathematics, University of California,
Santa Barbara, CA 93106-6105, U.S.A.}
\addresseshere

\end{document}